\documentclass[conference]{IEEEtran}
\IEEEoverridecommandlockouts
\usepackage{cite}
\usepackage{amsmath,amssymb,amsfonts}
\usepackage{algorithmic}
\usepackage{graphicx}
\usepackage{textcomp}
\usepackage{xcolor}

\usepackage{graphicx}
\usepackage{epstopdf}
\usepackage{multicol}
\usepackage{color}
\usepackage{array}
\usepackage{cite}
\usepackage{algorithm}
\usepackage{algorithmic}
\usepackage{multirow}
\usepackage{amsthm,amsmath}
\usepackage{mathrsfs}

\usepackage{bm}
\usepackage{courier}
\newtheorem{lemma}{Lemma}

\def\BibTeX{{\rm B\kern-.05em{\sc i\kern-.025em b}\kern-.08em
    T\kern-.1667em\lower.7ex\hbox{E}\kern-.125emX}}
\begin{document}
\title{Energy-Efficient Resource Allocation in a Multi-UAV-Aided NOMA Network
}
\author{\IEEEauthorblockN{Xing Xi\IEEEauthorrefmark{1}, Xianbin Cao\IEEEauthorrefmark{1}, Peng Yang\IEEEauthorrefmark{2}, Jingxuan Chen\IEEEauthorrefmark{1}, and Dapeng Wu\IEEEauthorrefmark{3}}

\IEEEauthorblockA{\IEEEauthorrefmark{1}Beihang University, Beijing, China, email: \{xixing,xbcao,chenjingxuan\}@buaa.edu.cn}
\IEEEauthorblockA{\IEEEauthorrefmark{2}Singapore University of Technology and Design, Singapore, email: peng\_yang@sutd.edu.sg}
\IEEEauthorblockA{\IEEEauthorrefmark{3}University of Florida, Gainesville, FL 32611, USA, email: dpwu@ufl.edu}
}

\maketitle

%\thanks{
%% This work was supported by the National Natural Science Foundation of China (NSFC) under grant 91738301 and 61827901.
%\emph{Corresponding author: Xianbin Cao.}
%
%X. Xi, X. Cao, and J. Chen are with School of Electronic and Information Engineering, Beihang University, Beijing, China. They are also with Key Laboratory of Advanced Technology of Near Space Information System (Beihang University), Ministry of Industry and Information Technology of China.
%
%P. Yang and T. Quek are with the Information Systems Technology and Design Pillar, Singapore University of Technology and Design, Singapore 487372.
%
%D. Wu is with Department of Electrical and Computer Engineering, University of Florida, Gainesville, FL 32611, USA. }

\maketitle

\begin{abstract}
This paper is concerned with the resource allocation in a multi-unmanned aerial vehicle (UAV)-aided network for providing enhanced mobile broadband (eMBB) services for user equipments. Different from most of the existing network resource allocation approaches, we investigate a joint non-orthogonal user association, subchannel allocation and power control problem. The objective of the problem is to maximize the network energy efficiency under the constraints on user equipments' quality of service, UAVs' network capacity and power consumption. We formulate the energy efficiency maximization problem as a challenging mixed-integer non-convex programming problem.
To alleviate this problem, we first decompose the original problem into two subproblems, namely, an integer non-linear user association and subchannel allocation subproblem and a non-convex power control subproblem. We then design a two-stage approximation strategy to handle the non-linearity of the user association and subchannel allocation subproblem and exploit a successive convex approximation approach to tackle the non-convexity of the power control subproblem. Based on the derived results, we develop an iterative algorithm with provable convergence to mitigate the original problem. Simulation results show that our proposed framework can improve energy efficiency compared with several benchmark algorithms.
\end{abstract}


% \fi
\IEEEpeerreviewmaketitle

\section{Introduction}
Enhanced mobile broadband (eMBB) has been identified as one of the three major services of 5G wireless networks \cite{series2015imt}. To provide high-quality eMBB services, which have high transmission rate requirements, the network capacity of the infrastructure should be robust. However, when network congestion or network failure caused by flash crowd traffic or infrastructure malfunction occurs in an area, terrestrial eMBB users may suffer from communication service interruption.

%of a communication service area may suffer from network failure or network congestion due to infrastructure malfunction, flash crowd traffic or remote areas.
% eMBB已经被定义为5G无线网络的三大服务之一。eMBB服务具有高的传输速率需求，其需要可靠的网络容量的保障。然而，当一个区域出现短时突发业务时或基础设施损毁时，eMBB的QoS 需求可能会经历服务中断由于网络拥塞或者网络失效。 用户可能 可能会经历服务中断由于network failure or network congestion (e.g., infrastructure malfunction, flash crowd traffic, and remote areas。
%eMBB应用具有严格的QoS需求在吞吐量方面，然而，一个区域可能会经历communication service interruption在case of network failure or network congestion (e.g., infrastructure malfunction, flash crowd traffic, and remote areas)。一个可行的方法是

A promising solution to alleviate the effect of network congestion or network failure is the utilization of unmanned aerial vehicle (UAV) base stations (i.e., low-altitude UAVs equipped with transceivers), which can support fast communication service recovery or even network performance enhancement \cite{xilouris2018uav}. Motivated by these advantages, UAV-aided communications are gradually attracting the attention of researchers.

Recent works on the UAV-aided communications mainly focus on network resource allocation. For example,
%Xi et al. investigated a joint user association and UAV location optimization problem for a single-UAV-aided network to maximize users’ total rate \cite{xi2019joint}.
Zhang et al. considered a multi-UAV-aided network and studied a joint subchannel allocation and UAV speed optimization problem to improve the uplink sum rate of the network \cite{zhang2019cellular}. Cui et al. investigated a dynamic resource allocation problem of a multi-UAV network to maximize long-term rewards. They proposed a multi-agent reinforcement learning-based algorithm to find the optimal strategy on joint user, subchannel and power level selection \cite{cui2019multi}.

However, the above works \cite{zhang2019cellular,cui2019multi} are all based on orthogonal multiple access (OMA) techniques. To further improve the utilization efficiency of network resources, non-orthogonal multiple access (NOMA) techniques have been studied for the UAV-aided communications. For example, Zhao et al. investigated a joint user scheduling, UAV trajectory and NOMA precoding problem for a UAV-aided NOMA network to maximize users' sum rate\cite{zhao2019joint}. Tang et al. studied a joint placement design, admission control, and power allocation problem for a heavy-loaded UAV-aided NOMA network to maximize the number of served users \cite{tang2019joint}. Nevertheless, the works in \cite{zhao2019joint,tang2019joint} considered single-UAV communications, which have the disadvantages of limited service capability and poor robustness compared to multi-UAV communications.
As a result, Duan et al. considered resource allocation for a multi-UAV-aided NOMA uplink network and jointly optimized subchannel allocation, transmit power, and UAVs' heights to improve the system capacity \cite{duan2019resource}. However, they adopted the classic K-mean clustering method to associate UAVs and users, which had low resource utilization.
%Meanwhile, compared with terrestrial networks, the capacity of UAV networks is more limited, and the classic K-mean clustering method for user association may not be able to achieve load balance among UAVs, which may lead to the network congestion.
Meanwhile, the classic K-mean clustering method could not achieve load balance among UAVs. Compared with terrestrial networks, the
capacity of UAV networks is stringently limited. Therefore, the proposed resource allocation algorithm in \cite{duan2019resource} might lead to UAV network congestion.

To improve the resource utilization and achieve load balance in a multi-UAV-aided NOMA downlink network, we investigate a joint non-orthogonal user association, subchannel allocation and power control problem in this paper. The main contributions are summarized as follows:
\begin{itemize}
\item  We formulate a joint non-orthogonal resource allocation optimization problem aiming at maximizing the network energy efficiency under the constraints on  quality of service (QoS) requirements, network capacity, and power consumption.
\item  The formulated problem is confirmed to be a challenging mixed-integer non-convex programming problem. To alleviate this problem, we decompose it into two separated subproblems, namely, an integer non-linear user association and subchannel allocation subproblem, and a non-convex power control subproblem.
\item  We then design a two-stage approximation strategy to handle the non-linearity of the user association and subchannel allocation subproblem and exploit a successive convex approximation (SCA) approach to tackle the non-convexity of the power control subproblem. Then an iterative algorithm with provable convergence is proposed to alternatively optimize the above two subproblems.
\end{itemize}

%\begin{itemize}
%\item We consider a NOMA-based downlink communication model in a multi-UAV-aided network.
%\item Based on this model, we formulate a joint non-orthogonal network resource allocation optimization problem with a goal of maximizing the system energy efficiency under the constraints on QoS requirements, network capacity, and power consumption. This energy efficiency maximization problem is confirmed to be a mixed-integer non-convex programming (MINCP) problem, which is challenging to mitigate.
%\item To alleviate the formulated energy efficiency maximization problem, we first decompose it into two separated subproblems, namely the optimization of subchannel allocation and the optimization of power control. We design a two-stage optimization method to solve the subchannel allocation subproblem and exploit a successive convex approximation approach to tackle the non-convexity of the power control subproblem. Then an iterative solution algorithm is proposed by alternatively optimizing these two subproblems.
%%\item Simulation results demonstrate that our proposed algorithm can achieve the performance (energy efficiency) gain compared with other benchmark algorithms.
%\end{itemize}

The rest of this paper is organized as follows: We present the system model and the problem formulation in Section \uppercase\expandafter{\romannumeral2}. We develop the problem solution for the formulated problem in Section \uppercase\expandafter{\romannumeral3}. Section \uppercase\expandafter{\romannumeral4} shows our simulation results and Section \uppercase\expandafter{\romannumeral5} concludes this paper.

\section{System Model and Problem Formulation}

\subsection{System Model}
In this paper, we consider a NOMA-based downlink communication scenario. In this scenario, multiple UAV base stations (UBSs) are deployed to assist a macro base station (MBS) to provide eMBB services for a collection of congested terrestrial user equipments (UEs) which cannot be served by the MBS in a geographical area.
%All UBSs connect to a centralized network operator by wireless fronthaul links, and the operator decides whether to accept or reject the slice requests, so that network slicing can be implemented smoothly.
Denote the set of UBSs and the set of UEs by $\mathcal{J}=\{1,2,\ldots ,N_{d}\}$ and $\mathcal{I}=\{1,2,\ldots ,N_{u}\}$ respectively. We consider that the locations of all UBSs and UEs are fixed and known, and all UBSs are deployed at the same altitude $H$. For simplicity, we ignore the height of the MBS and the UEs. Meanwhile, this paper considers a frequency division multiple access (FDMA) communication system. The total channel bandwidth is $W$ and is equally divided into $N_s$ orthogonal subchannels, denoted by $\mathcal{N}=\{1,2,\ldots ,N_{s}\}$. For convenience of description, we denote the subchannel $n$ of UBS $j$ as ${\mathcal{SC}_{jn}}$.
%Denote ${\mathcal{SC}_{jn}}$ as the subchannel $n$ of UBS $j$.
%Since multiple subchannels can simultaneously serve a UE, we regard each subchannel ${\mathcal{SC}_{jn}}$ as a network slice.
%In the network slicing system, a network slice is allocated with a subchannel ${\mathcal{SC}_{jn}}$.
Let ${{a}_{ijn}}$ be a binary variable indicating user association and subchannel allocation and let ${\mathcal{A}} = \{{{a}_{ijn}},\forall i,j,n\}$ denote the user association and subchannel allocation matrix. We set $a_{ijn}=1$ if the subchannel ${\mathcal{SC}_{jn}}$ is allocated to UE $i$;
%a slice request that slice ${\mathcal{SC}_{jn}}$ provides service to UE $i$ is accepted/admitted by the network operator;
otherwise, $a_{ijn}=0$. This paper investigates the optimization of joint user association, subchannel allocation and UBSs' transmit power control, and we assume that the transmit power of the MBS is fixed and known.
%as shown in Fig. 1
%\begin{figure}[!t]
%\centering
%\includegraphics[width=3.5in]{integrated_network.eps}
%% where an .eps filename suffix will be assumed under latex,
%% and a .pdf suffix will be assumed for pdflatex; or what has been declared
%% via \DeclareGraphicsExtensions.
%\caption{A downlink communication scenario in a multi-UAV-aided network.
%}
%\label{integrated_network}
%\end{figure}

Denote the horizontal location of UBS $j$ and the location of UE $i$ by $\bm{x}_{j}^{d}$ and $\bm{x}_{i}^{u}$ respectively.
This paper leverages the air-to-ground (ATG) propagation model \cite{al2014optimal} to obtain the channel gain from UBS $j$ to UE $i$ on the subchannel $n$, denoted by ${{h}_{ijn}}$. For the ATG link, each UE has a line-of-sight (LoS) connection with a UBS with a specific probability. The LoS probability relies on the environment (e.g., rural, suburban, urban and dense urban), the locations of the UBS and the UE, and can be expressed as %${{P}_{LoS}}(H,d_{ij}^{h})=\frac{1}{1+{{\alpha }_{1}}exp(-{{\alpha }_{2}}({{\theta }_{ij}}-{{\alpha }_{1}}))}$,
\begin{equation}\label{PLoS}
{{P}_{LoS}}(H,d_{ij}^{h})=\frac{1}{1+{{\alpha }_{1}}exp(-{{\alpha }_{2}}({{\theta }_{ij}}-{{\alpha }_{1}}))},
\end{equation}
%, i.e., $d_{ij}^{h}={{\left\| \bm{x}_{i}^{u}-\bm{x}_{j}^{d} \right\|}_{2}}$
where ${{\alpha _1}}$ and ${{\alpha _2}}$ are constant values depending on the environment, ${{\theta }_{ij}}=\frac{180}{\pi }\times \arctan (\frac{H}{d_{ij}^{h}})$ is the elevation angle of UE $i$ towards UBS $j$, and ${d_{ij}^h}$ is the horizontal distance between UBS $j$ and UE $i$, i.e., $d_{ij}^{h}={{\left\| \bm{x}_{i}^{u}-\bm{x}_{j}^{d} \right\|}_{2}}$. Also, the non-line-of-sight (NLoS) probability is ${{P}_{NLoS}}(H,d_{ij}^{h})=1-{{P}_{LoS}}(H,d_{ij}^{h})$. Thus, the channel gain from UBS $j$ to UE $i$ on the subchannel $n$ is
%${h_{ijn}} = \frac{{g_{ijn}^{Tx}g_{ijn}^{Rx}{\varsigma ^2}}}{{16{\pi ^2}{{\left( {\frac{{{d_{ij}}}}{{{d_0}}}} \right)}^2}}}{10^{ - \frac{{{P_{LoS}}(H,d_{ij}^h)\eta _{LoS}^{dB} + {P_{NLoS}}(H,d_{ij}^h)\eta _{NLoS}^{dB}}}{{10}}}}$,
\begin{equation}\label{Channel_Gain_ATG}
{h_{ijn}} = \frac{{g_{ijn}^{Tx}g_{ijn}^{Rx}{\varsigma ^2}}}{{16{\pi ^2}{{\left( {\frac{{{d_{ij}}}}{{{d_0}}}} \right)}^2}}}{10^{ - \frac{{{P_{LoS}}(H,d_{ij}^h)\eta _{LoS}^{dB} + {P_{NLoS}}(H,d_{ij}^h)\eta _{NLoS}^{dB}}}{{10}}}},
\end{equation}
%\begin{equation}\label{Channel_Gain_ATG}
%{{h}_{ijn}}=\frac{g_{ijn}^{Tx}g_{ijn}^{Rx}{{\varsigma }^{2}}}{16{{\pi }^{2}}{{\left( \frac{{{d}_{ij}}}{{{d}_{0}}} \right)}^{2}}}{{\left( {{\eta }_{LoS}} \right)}^{{{P}_{LoS}}(h,d_{ij}^{h})}}{{\left( {{\eta }_{NLoS}} \right)}^{{{P}_{NLoS}}(h,d_{ij}^{h})}}
%\end{equation}
where $g_{ijn}^{Tx}$ and $g_{ijn}^{Rx}$ are the transmit and receive antenna gains from UBS $j$ to UE $i$ on the subchannel $n$. $\varsigma = c/{f_c}$ is the carrier wavelength, where $c$ is the speed of light and ${f_c}$ is the carrier frequency. ${{d}_{ij}}=\sqrt{{{\left( d_{ij}^{h} \right)}^{2}}+{{H}^{2}}}$ is the distance between UBS $j$ and UE $i$ and ${{d}_{0}}$ is a far field reference distance. $\eta _{LoS}^{dB}$ (in dB) and $\eta _{NLoS}^{dB}$ (in dB) represent the excessive propagation losses corresponding to the LoS and NLoS connections respectively, which depend on the environment.

Denote the location of the MBS by $\bm{x}^{M}$.
This paper leverages the propagation path loss model \cite{rouphael2009rf} to obtain the channel gain from the MBS to UE $i$ on the subchannel $n$, denoted by $h_{in}^{M}$. Thus, the channel gain $h_{in}^{M}$ is %$h_{in}^M = \frac{{g_{in}^{MTx}g_{in}^{MRx}{\varsigma ^2}}}{{16{\pi ^2}{{\left( {\frac{{d_i^M}}{{{d_0}}}} \right)}^\eta }}}$,
\begin{equation}\label{Channel_Gain_GTG}
h_{in}^M = \frac{{g_{in}^{MTx}g_{in}^{MRx}{\varsigma ^2}}}{{16{\pi ^2}{{\left( {\frac{{d_i^M}}{{{d_0}}}} \right)}^\eta }}},
\end{equation}
where $g_{in}^{MTx}$ and $g_{in}^{MRx}$ are the transmit and receive antenna gains from the MBS to UE $i$ on the subchannel $n$, $d_{i}^{M}={{\left\| \bm{x}_{i}^{u}-{\bm{x}^{M}} \right\|}_{2}}$ is the distance between the MBS and UE $i$, and $\eta$ is the path loss exponent ($\eta  \in [2,6]$).
%$d_{i}^{M}={{\left\| \bm{x}_{i}^{u}-{\bm{x}^{M}} \right\|}_{2}}$

In the NOMA-based downlink system, the successive interference cancellation (SIC) technique is adopted at the receiver to eliminate the interference from other UEs served by the same subchannel ${\mathcal{SC}_{jn}}$ in a certain decoding order \cite{zhang2018energy,zhao2019joint}.
We assume that the UE with higher channel gain can decode the signals of the other UEs with worse channel gain served by the same subchannel ${\mathcal{SC}_{jn}}$, and the transmit power allocated to the former is not more than that of the latter.
%Besides, considering the implementation complexity and decoding complexity of SIC and the complexity of resource allocation algorithm, we investigate a simple case where each ${\mathcal{SC}_{jn}}$ can be allocated to at most two UEs. Many researches on NOMA also consider this simple case \cite{xiao2018joint}. Thus, we have
Owing to the high implementation complexity and decoding complexity of SIC and the high complexity of resource allocation algorithms, like \cite{xiao2018joint}, we investigate the case that each ${\mathcal{SC}_{jn}}$ can be allocated to at most two UEs. In consequence, we have
%The higher number of users that can be allocated for each ${\mathcal{SC}_{jn}}$ may lead to an increase in system performance, but it may also lead to intolerable complexity.
\begin{equation}\label{C1}
\text{C1: }\sum\nolimits_{i \in {\mathcal I}} {{a_{ijn}}}  \le 2,\forall j \in {\mathcal J},n \in {\mathcal N},
\end{equation}
\begin{equation}\label{C2}
\text{C2: }{a_{ijn}} \in \{ 0,1\} ,\forall i \in {\mathcal I},j \in {\mathcal J},n \in {\mathcal N}.
\end{equation}

Considering the number of UEs served by the subchannel ${\mathcal{SC}_{jn}}$, we calculate the received signal-to-interference-plus-noise ratio (SINR) in the following two cases.

Case 1: When ${\mathcal{SC}_{jn}}$ is allocated to only one UE $i$, we name UE ${i}$ as a primary UE on ${\mathcal{SC}_{jn}}$. Then, the received SINR of the primary UE $i$ on ${\mathcal{SC}_{jn}}$ is
%from UBS $j$ on the subchannel $n$ is
\begin{equation}\label{SINR_Case1}
{\gamma _{ijn}} = \frac{{{p_{1,jn}}{h_{ijn}}}}{{\sum\limits_{k \ne j,k \in {\mathcal J}} {{p_{kn}}{h_{ikn}}}  + p_n^Mh_{in}^M + \sigma _n^2}}.
\end{equation}

Case 2: When ${\mathcal{SC}_{jn}}$ is allocated to two UEs ${{i}_{1}}$ and ${{i}_{2}}$ with ${{h}_{{{i}_{1}}jn}}\!>\!{{h}_{{{i}_{2}}jn}}$, i.e., UE ${{i}_{1}}$ can eliminate the interference of UE ${{i}_{2}}$ on ${\mathcal{SC}_{jn}}$, we name UE ${{i}_{1}}$ and UE ${{i}_{2}}$ as a primary UE and a secondary UE on ${\mathcal{SC}_{jn}}$ respectively. Then, the received SINRs of the primary UE $i_1$ and the secondary UE $i_2$ on ${\mathcal{SC}_{jn}}$ are
%from UBS $j$ on the subchannel $n$ are
\begin{equation}\label{SINR_Case21}
{\gamma _{{i_1}jn}} = \frac{{{p_{1,jn}}{h_{{i_1}jn}}}}{{\sum\limits_{k \ne j,k \in {\mathcal J}} {{p_{kn}}{h_{{i_1}kn}}}  + p_n^Mh_{{i_1}n}^M + \sigma _n^2}},
\end{equation}
\begin{equation}\label{SINR_Case22}
{\gamma _{{i_2}jn}} = \frac{{{{p}_{2,jn}}{h_{{i_2}jn}}}}{{{p_{1,jn}}{h_{{i_2}jn}}  + \!\!\! \sum\limits_{k \ne j,k \in {\mathcal J}} \!\!\! {{p_{kn}}{h_{{i_2}kn}}}  + p_n^Mh_{{i_2}n}^M  +  \sigma _n^2}},
\end{equation}
where ${{p}_{1,jn}}$ and ${{p}_{2,jn}}$ are the transmit powers allocated to the primary UE and the secondary UE on ${\mathcal{SC}_{jn}}$ respectively, ${{p}_{jn}}={{p}_{1,jn}}+{{p}_{2,jn}}$ is the total transmit power on ${\mathcal{SC}_{jn}}$, ${p_n^M}$ is the transmit power of the MBS on the subchannel $n$, and $\sigma _{n}^{2}$ represents the additive white Gaussian noise on the subchannel $n$. Let ${\mathcal{S}_{c2}}=\{(j,n)|\sum\nolimits_{i\in \mathcal{I}}{a_{ijn}}=2\}$ represent the index set corresponding to the case 2. Referring to the definition of ${{p}_{1,jn}}$ and ${{p}_{2,jn}}$ and the assumption of SIC, ${{p}_{1,jn}}$ and ${{p}_{2,jn}}$ satisfy the following constraints
\begin{equation}\label{C3}
\text{C3: } {p_{2,jn}} = 0,\forall (j,n) \notin {\mathcal{S}_{c2}},
\end{equation}
\begin{equation}\label{C4}
\text{C4: }{p_{1,jn}} \le {p_{2,jn}},\forall (j,n) \in {\mathcal{S}_{c2}}.
\end{equation}

Let $\mathcal{P}=\{{{p}_{1,jn}},{{p}_{2,jn}},\forall j,n\}$ denote the transmit power matrix. According to the Shannon capacity, the achievable data transfer rate of UE $i$ on ${\mathcal{SC}_{jn}}$ is
\begin{equation}\label{r_ijn}
{r_{ijn}} = \frac{W}{N_s}{\log _2}\left( {1 + {\gamma _{ijn}}} \right).
\end{equation}

We define UEs' different QoS requirements by their achievable data transfer rate (Mb/s). Denote the achievable data transfer rate and the minimum required data transfer rate of UE $i$ as ${{R}_{i}}$ and $R_i^{\min}$ respectively. Thus, we have
\begin{equation}\label{C5}
\text{C5: }{R_i} = \sum\limits_{j \in {\mathcal J}} {\sum\limits_{n \in {\mathcal N}} {{a_{ijn}}{r_{ijn}}} }  \ge R_i^{\min },\forall i \in {\mathcal I}.
\end{equation}

Besides, UBS $j$ needs to receive data that will be forwarded to UEs from a ground station via an uplink with limited network capacity (Mb/s). In this paper, we regard it as the maximum network capacity of UBS $j$, denoted by $C_{j}^{\max }$. Thus, we have
\begin{equation}\label{C6}
\text{C6: }\sum\limits_{i \in {\mathcal I}} {\sum\limits_{n \in {\mathcal N}} {{a_{ijn}}{r_{ijn}}} }  \le C_j^{\max },\forall j \in {\mathcal J}.
\end{equation}
% Xi2019Efficicent

Next, let ${{p}_{j}}$, $p_{j}^{c}$ and $p_{j}^{\max }$ denote the transmit power, the circuit power and the maximum power consumption limit of UBS $j$. Thus, we have
\begin{equation}\label{C7}
\text{C7: }{p_j} = \sum\nolimits_{n \in {\mathcal N}} {\left( {{p_{1,jn}} + {p_{2,jn}}} \right)}, \forall j \in {\mathcal J},
\end{equation}
\begin{equation}\label{C8}
\text{C8: }{p_j}  + p_j^c \le p_j^{\max },\forall j \in {\mathcal J},
\end{equation}
\begin{equation}\label{C9}
\text{C9: }{p_{1,jn}} \ge 0,{p_{2,jn}} \ge 0,\forall j \in {\mathcal J},n \in {\mathcal N}.
\end{equation}

We denote the energy efficiency as ${{f }_{EE}}$. Considering the fairness of service among UEs and the fairness of power consumption among UBSs, we define ${{f }_{EE}}$ as the ratio of the product of the number of UEs and the minimum achievable data transfer rate among all UEs and the product of the number of UBSs and the maximum power consumption among all UBSs. As such, the objective function can be written as
\begin{equation}\label{EE}
{f _{EE}} = \frac{ N_{u} \cdot{\mathop {\min }\nolimits_{i \in {\mathcal I}} {R_i}}}{ N_{d} \cdot {\mathop {\max }\nolimits_{j \in {\mathcal J}} ({p_j}+{p_j^c}) }},
\end{equation}
where ${ N_{u} \cdot{\mathop {\min }\nolimits_{i \in {\mathcal I}} {R_i}}}$ represents the lower bound of the total achievable data transfer rate of all UEs and ${ N_{d} \cdot {\mathop {\max }\nolimits_{j \in {\mathcal J}} ({p_j}+{p_j^c}) }}$ represents the upper bound of the total power consumption of all UBSs.

\subsection{Problem Formulation}
Considering all constraints and the objective function mentioned above, we can formulate the joint association, subchannel and power optimization problem as
\begin{equation}\label{P1o}
\begin{array}{l}
\mathop {\max }\limits_{{\mathcal A},{\mathcal P}} \ {f _{EE}} = \frac{ N_{u}}{ N_{d}} \cdot \frac{{\mathop {\min }\nolimits_{i \in {\mathcal I}} {R_i}}}{{\mathop {\max }\nolimits_{j \in {\mathcal J}} ({p_j}+{p_j^c}) }}\\
s.t. {\text{ (\ref{C1}), (\ref{C2}), (\ref{C3}), (\ref{C4}), (\ref{C5}), (\ref{C6}), (\ref{C7}), (\ref{C8}), (\ref{C9})}}.
\end{array}
\end{equation}

Since both $N_{u}$ and $N_{d}$ are given constants, solving (\ref{P1o}) is equivalent to the solution of the following (\ref{P1})
\begin{equation}\label{P1}
\begin{array}{l}
\mathop {\max }\limits_{{\mathcal A},{\mathcal P}} \ {\eta _{EE}} = \frac{{\mathop {\min }\nolimits_{i \in {\mathcal I}} {R_i}}}{{\mathop {\max }\nolimits_{j \in {\mathcal J}} ({p_j}+{p_j^c}) }}\\
s.t. {\text{ (\ref{C1}), (\ref{C2}), (\ref{C3}), (\ref{C4}), (\ref{C5}), (\ref{C6}), (\ref{C7}), (\ref{C8}), (\ref{C9})}}.
\end{array}
\end{equation}

We define the optimal $\eta _{EE}^{*}$ as
\begin{equation}\label{EE_Opt}
\eta _{EE}^* = \frac{{\mathop {\min }\nolimits_{i \in {\mathcal I}} {R_i}({{\mathcal A}^*},{{\mathcal P}^*})}}{{\mathop {\max }\nolimits_{j \in {\mathcal J}} ({p_j}({{\mathcal P}^*})+{p_j^c})}},
\end{equation}
where ${{\mathcal{A}}^{*}}$ and ${{\mathcal{P}}^{*}}$ denote the optimal user association and subchannel allocation and the optimal transmit power when yielding $\eta _{EE}^{*}$.

\begin{lemma}\label{lem:1}
\rm {$\eta _{EE}^{*}$ can be achieved if and only if \cite{zhang2018energy}
}
\begin{equation}\label{lem1}
\begin{array}{l}
\mathop {\max }\limits_{{\mathcal A},{\mathcal P}} \Big( {\mathop {\min }\limits_{i \in {\mathcal I}} {R_i}({\mathcal A},{\mathcal P})} \Big) - \eta _{EE}^*\Big( {\mathop {\max }\limits_{j \in {\mathcal J}} ({p_j}({\mathcal P})+{p_j^c})} \Big)\\
 = \Big( {\mathop {\min }\limits_{i \in {\mathcal I}} {R_i}({{\mathcal A}^*},{{\mathcal P}^*})} \Big) - \eta _{EE}^*\Big( {\mathop {\max }\limits_{j \in {\mathcal J}} ({p_j}({{\mathcal P}^*})+{p_j^c})} \Big) = 0.
\end{array}
\end{equation}
\end{lemma}

\begin{proof}
A similar proof can be found in \cite{zhang2018energy}.
%Please refer to Appendix A.
%Provided in Appendix A in \cite{xi2020non}.
%Due to the space limitation, we omit the proof of Lemma 1. A similar proof can be found in \cite{zhang2018energy}.
%Due to the space limitation, we omit the proof of Lemma 1. A similar proof can be found in \cite{zhang2018energy}.
\end{proof}

According to Lemma \ref{lem:1}, we can transform the objective function in (\ref{P1}) into a subtractive form, and the problem (\ref{P1}) can be rewritten as
\begin{equation}\label{P2}
\begin{array}{l}
\mathop {\max }\limits_{{\mathcal A},{\mathcal P}} {\mkern 1mu} {\rm{ }}\Big( {\mathop {\min }\limits_{i \in {\mathcal I}} {R_i}} \Big) - {\eta _{EE}}\Big( {\mathop {\max }\limits_{j \in {\mathcal J}} ({p_j}+{p_j^c})} \Big)\\
s.t. {\text{ (\ref{C1}), (\ref{C2}), (\ref{C3}), (\ref{C4}), (\ref{C5}), (\ref{C6}), (\ref{C7}), (\ref{C8}), (\ref{C9})}}.
\end{array}
\end{equation}

In the problem (\ref{P2}), (\ref{C1}), (\ref{C2}), (\ref{C5}) and (\ref{C6}) involve binary variables $\{{{a}_{ijn}}\}$. Further, even if $\{{{a}_{ijn}}\}$ are fixed, (\ref{C5}) and (\ref{C6}) are not convex constraints. Therefore, (\ref{P2}) is a mixed-integer non-convex programming problem, which is indeterminable or NP-hard \cite{lee2011mixed} and challenging to be solved.
Besides, ${\mathcal A}$ and ${\mathcal P}$ are coupled in the objective function and the constraints (\ref{C5}) and (\ref{C6}), which increases the difficulty of mitigating (\ref{P2}). Fortunately, we observe that the complexity of (\ref{P2}) may be weakened if ${\mathcal A}$ and ${\mathcal P}$ can be decoupled.
%在此之前写上一个motivation（为什么要进行子问题分解）：
%Besides, 我们观察到在有些constrain同时中，A与P是相互耦合(coupled)的，这增大了同时优化A 与P 的难度。若是能够解耦(decouple)A与P，那将降低问题的求解难度。根据这个key observation，我们就将原问题分解成了两个子问题。
Based on this crucial observation, we first decompose (\ref{P2}) into two separated subproblems, namely, association and subchannel optimization with fixed transmit power and power optimization with fixed user association and subchannel allocation. Based on the solutions of the above two subproblems, we then develop an iterative algorithm for (\ref{P2}) to alternatively optimize these two subproblems. The detailed procedures are described in the following section.

\section{Problem solution}

\subsection{User Association and Subchannel Allocation}
For any given transmit power $\mathcal{P}$, this subsection considers the subproblem of (\ref{P2}) of user association and subchannel allocation. By introducing auxiliary variables ${{\eta }_{R}}$ and $\{{{\eta }_{i}},\forall i\in \mathcal{I}\}$, the user association and subchannel allocation subproblem can be formulated as
\begin{equation}\label{P3}
\begin{array}{l}
\mathop {\max }\limits_{{\mathcal A},{{\eta }_{R}},\{{{\eta }_{i}}\} } \ {\eta _R}\\
%-\rho \sum\limits_{i \in {\mathcal I}} \sum\limits_{j \in {\mathcal J}} {\sum\limits_{n \in {\mathcal N}} {a_{ijn}} } \\
s.t.\\
{\eta _i}  = \sum\limits_{j \in {\mathcal J}} {\sum\limits_{n \in {\mathcal N}} {{a_{ijn}}{r_{ijn}({\mathcal A})}} },\forall i \in {\mathcal I},\\
{\eta _i} \ge R_i^{\min },\forall i \in {\mathcal I},\\
{\eta _i} \ge {\eta _R},\forall i \in {\mathcal I},\\
{\text{ (\ref{C1}), (\ref{C2}) and (\ref{C6}) are satisfied,}}
%s.t\\
%\sum\nolimits_{i \in {\mathcal I}} {{a_{ijn}}}  \le 2,\forall j \in {\mathcal J},n \in {\mathcal N}\\
%\sum\limits_{j \in {\mathcal J}} {\sum\limits_{n \in {\mathcal N}} {{a_{ijn}}{r_{ijn}}} }  \ge {\eta _i},\forall i \in {\mathcal I}\\
%{\eta _i} \ge R_i^{\min },\forall i \in {\mathcal I}\\
%{\eta _i} \ge {\eta _R},\forall i \in {\mathcal I}\\
%\sum\limits_{i \in {\mathcal I}} {\sum\limits_{n \in {\mathcal N}} {{a_{ijn}}{r_{ijn}}} }  \le C_j^{\max},\forall j \in {\mathcal J}\\
%{a_{ijn}} \in \{ 0,1\} ,\forall i \in {\mathcal I},j \in {\mathcal J},n \in {\mathcal N}
\end{array}
\end{equation}
where ${\eta_i}$ represents the achievable data transfer rate of UE $i$ and ${\eta_R}$ represents the minimum achievable data transfer rate among all UEs.
%where $\rho$ is a sufficiently small positive number, and $-\rho \sum\limits_{i \in {\mathcal I}} \sum\limits_{j \in {\mathcal J}} {\sum\limits_{n \in {\mathcal N}} {a_{ijn}} }$ is used to avoid meaningless subchannel allocation.

However, the above problem (\ref{P3}) is challenging to be solved since the achievable data transfer rate ${{r}_{ijn}({\mathcal A})}$ is not a fixed value even with fixed transmit power $\mathcal{P}$. From (\ref{SINR_Case1}), (\ref{SINR_Case21}), (\ref{SINR_Case22}) and (\ref{r_ijn}), it can be observed that ${{r}_{ijn}({\mathcal A})}$ takes different values depending on whether UE $i$ is a primary UE or a secondary UE on ${\mathcal{SC}_{jn}}$. Therefore, the value of ${{r}_{ijn}({\mathcal A})}$ depends on $\mathcal{A}$ with the fixed $\mathcal{P}$. According to this key observation, we propose a two-stage approximation strategy to further decompose the problem (\ref{P3}) into two integer linear programming (ILP) problems, both of which can be solved efficiently by existing optimization tools such as MOSEK \cite{mosek2019}. The detailed procedures are described as follows.
%MOSEK ApS, “MOSEK optimization toolbox for MATLAB 8.1.0.67,” 2018. [Online]. Available: https://docs.mosek.com/8.1/toolbox/index.html
%M. Grant and S. Boyd. (2016). CVX: MATLAB Software for Disciplined Convex Programming. [Online]. Available: http://cvxr.com/cvx

\subsubsection{Primary User Association and Subchannel Allocation}
%For the given transmit power $\mathcal{P}$, this subsection considers the primary association. Particularly,
We assume that each ${\mathcal{SC}_{jn}}$ can be allocated to at most one UE at the primary user association and subchannel allocation stage (hereinafter referred to as the primary stage). Therefore, each UE $i$ can be regarded as a primary UE on ${\mathcal{SC}_{jn}}$ at this stage, and thus we can formulate the primary user association and subchannel allocation subproblem as the following ILP problem
\begin{equation}\label{P4}
\begin{array}{l}
\mathop {\max }\limits_{{\mathcal A},{{\eta }_{R}},\{{{\eta }_{i}}\} } \ {\eta _R}\\
%-\rho \sum\limits_{i \in {\mathcal I}} \sum\limits_{j \in {\mathcal J}} {\sum\limits_{n \in {\mathcal N}} {a_{ijn}} }\\
s.t\\
\sum\nolimits_{i \in {\mathcal I}} {{a_{ijn}}}  \le 1,\forall j \in {\mathcal J},n \in {\mathcal N},\\
{\eta _i}  = \sum\limits_{j \in {\mathcal J}} {\sum\limits_{n \in {\mathcal N}} {{a_{ijn}}{r_{ijn}^p}} },\forall i \in {\mathcal I},\\
{\eta _i} \ge R_i^{\min },\forall i \in {\mathcal I},\\
{\eta _i} \ge {\eta _R},\forall i \in {\mathcal I},\\
\sum\limits_{i \in {\mathcal I}} {\sum\limits_{n \in {\mathcal N}} {{a_{ijn}}{r_{ijn}^p}} }  \le C_j^{\max},\forall j \in {\mathcal J},\\
{a_{ijn}} \in \{ 0,1\} ,\forall i \in {\mathcal I},j \in {\mathcal J},n \in {\mathcal N},
\end{array}
\end{equation}
where ${r_{ijn}^p} = \frac{W}{N_s}{\log _2}\Big( {1 + \frac{{{p_{1,jn}}{h_{ijn}}}}{{\sum\limits_{k \ne j,j \in {\mathcal J}} {{p_{kn}}{h_{ikn}}}  + p_n^Mh_{in}^M + \sigma _n^2}}} \Big)$.

Let $\{a_{ijn}^{p*}\}$ denote the solution of (\ref{P4}) and ${{\mathcal{S}}_{p*}}=\{(i,j,n)|a_{ijn}^{p*}=1\}$ denote the index set of the user association and subchannel allocation determined at the primary stage. In addition, for the convenience of the description of the following secondary user association and subchannel allocation, we let $R_{i}^{p}=\sum\limits_{j\in \mathcal{J}}{\sum\limits_{n\in \mathcal{N}}{a_{ijn}^{p*}{{r}_{ijn}^p}}}$ and $C_{j}^{p}=\sum\limits_{i\in \mathcal{I}}{\sum\limits_{n\in \mathcal{N}}{a_{ijn}^{p*}{{r}_{ijn}^p}}}$ represent the achievable data transfer rate of UE $i$ and the total data transfer rate of UBS $j$ at the primary stage respectively.

\subsubsection{Secondary User Association and Subchannel Allocation}
%For the given transmit power $\mathcal{P}$ and the primary association $\{a_{ijn}^{p*}\}$, this subsection considers the secondary association.
Similarly, we assume that each ${\mathcal{SC}_{jn}}$ can be allocated to at most one UE at the secondary user association and subchannel allocation stage (hereinafter referred to as the secondary stage). Particularly, based on the primary stage, we can calculate the achievable data transfer rate $r_{ijn}$ at the secondary stage in the following two cases.

Case 1: For each $\mathcal{SC}_{jn}$, if ${\mathcal{SC}_{jn}}$ is not allocated to any UE at the primary stage, then UE $i$ is a primary UE on $\mathcal{SC}_{jn}$ at the secondary stage. Let ${{\mathcal{S}}_{u}}=\{(i,j,n)|\sum\nolimits_{k\in \mathcal{I}}{a_{kjn}^{p*}}=0\}$ represent the index set of the user association and subchannel allocation corresponding to this case. Thus, for each $(i,j,n)\in {{\mathcal{S}}_{u}}$, the achievable data transfer rate $r_{ijn}^{u}$ at the secondary stage is
\begin{equation}\label{r_su}
r_{ijn}^{u} = \frac{W}{N_s}{\log _2}\Big( {1 + \frac{{{p_{1,jn}}{h_{ijn}}}}{{\sum\limits_{k \ne j,k \in {\mathcal J}} {{p_{kn}}{h_{ikn}}}  + p_n^Mh_{in}^M + \sigma _n^2}}} \Big).
\end{equation}

Case 2: If ${\mathcal{SC}_{jn}}$ is allocated to a UE at the primary stage, then we denote this UE as ${{i}_{p}}$, i.e., $({{i}_{p}},j,n)\in {{\mathcal{S}}_{p*}}$. When ${\mathcal{SC}_{jn}}$ is allocated to UE ${{i}_{p}}$ and UE $i$ at the primary and secondary stages respectively, one of the two UEs is a primary UE and the other is a secondary UE on $\mathcal{SC}_{jn}$, which is determined by the relative relationship of the two UEs' channel gains. If ${{h}_{ijn}}>{{h}_{{{i}_{p}}jn}}$, then UE $i$ and UE ${{i}_{p}}$ are the primary UE and the secondary UE on $\mathcal{SC}_{jn}$ respectively, and the achievable data transfer rate $r_{{{i}_{p}}jn}^{p}$ of UE ${{i}_{p}}$ at the primary stage will change. Let ${{\mathcal{S}}_{o1}}=\{(i,j,n)|{{h}_{ijn}}>{{h}_{{{i}_{p}}jn}},({{i}_{p}},j,n)\in {{\mathcal{S}}_{p*}}\}$ represent the index set of the user association and subchannel allocation corresponding to this case. Thus, for each $(i,j,n)\in {{\mathcal{S}}_{o1}}$, the achievable data transfer rate $r_{ijn}^{o1}$ at the secondary stage is
\begin{equation}\label{r_sn1}
r_{ijn}^{o1} = \frac{W}{N_s}{\log _2}\Big( {1 + \frac{{{p_{1,jn}}{h_{ijn}}}}{{\sum\limits_{k \ne j,k \in {\mathcal J}} {{p_{kn}}{h_{ikn}}}  + p_n^Mh_{in}^M + \sigma _n^2}}} \Big),
\end{equation}
and the change of the achievable data transfer rate $r_{{{i}_{p}}jn}^{p}$ at the primary stage is
\begin{equation}\label{delta_rp}
\begin{array}{l}
\Delta r_{{i_p}jn}^p =\\
\frac{W}{N_s}{\log _2}\Big( {1 + \frac{{{p_{2,jn}}{h_{{i_p}jn}}}}{{{p_{1,jn}}{h_{{i_p}jn}} + \sum\limits_{k \ne j,j \in {\mathcal J}} {{p_{kn}}{h_{{i_p}kn}}}  + p_n^Mh_{{i_p}n}^M + \sigma _n^2}}} \Big)\\
 - \frac{W}{N_s}{\log _2}\Big( {1 \!\! + \!\! \frac{{{p_{1,jn}}{h_{{i_p}jn}}}}{{\sum\limits_{k \ne j,j \in {\mathcal J}} {{p_{kn}}{h_{{i_p}kn}}}  + p_n^Mh_{{i_p}n}^M + \sigma _n^2}}}
 \Big).
\end{array}
\end{equation}

If ${{h}_{ijn}} < {{h}_{{{i}_{p}}jn}}$, then UE ${{i}_{p}}$ and UE $i$ are the primary UE and the secondary UE on $\mathcal{SC}_{jn}$ respectively. Let ${{\mathcal{S}}_{o2}}=\{(i,j,n)|{{h}_{ijn}}<{{h}_{{{i}_{p}}jn}},({{i}_{p}},j,n)\in {{\mathcal{S}}_{p*}}\}$ represent the index set of the user association and subchannel allocation corresponding to this case. Thus, for each $(i,j,n)\in {{\mathcal{S}}_{o2}}$, the achievable data transfer rate $r_{ijn}^{o2}$ at the secondary stage is
\begin{equation}\label{r_sn2}
\begin{array}{l}
r_{ijn}^{o2} =
\frac{W}{N_s}{\log _2}\Big( {1 + \frac{{{p_{2,jn}}{h_{ijn}}}}{{{p_{1,jn}}{h_{ijn}} + \sum\limits_{k \ne j,j \in {\mathcal J}} {{p_{kn}}{h_{ikn}}}  + p_n^Mh_{in}^M + \sigma _n^2}}} \Big).
\end{array}
\end{equation}

Let ${{\mathcal{S}}_{2}}={{\mathcal{S}}_{u}}\cup {{\mathcal{S}}_{o1}}\cup {{\mathcal{S}}_{o2}}$ represent the candidate index set of the feasible user association and subchannel allocation at the secondary stage, and let $\mathcal{S}_{JN}^{p*}(i)=\{(j,n)|(i,j,n)\in {{\mathcal{S}}_{p*}}\}$, ${{\mathcal{I}}_{o1}}(j,n)=\{i|(i,j,n)\in {{\mathcal{S}}_{o1}}\}$ and $\mathcal{S}_{IN}^{o1}(j)=\{(i,n)|(i,j,n)\in {{\mathcal{S}}_{o1}}\}$. For each $i\in {\mathcal I}$, the change of the achievable data transfer rate of UE $i$ at the primary stage is
\begin{equation}\label{delta_Rip}
\Delta R_i^p = \sum\limits_{(j,n) \in {\mathcal S}_{JN}^{p*}(i)} {\sum\limits_{k \in {{\mathcal I}_{o1}}(j,n)} {{a_{kjn}}} \Delta r_{ijn}^p}.
\end{equation}

For each ${j \in {\mathcal J}}$, the change of the total data transfer rate of UBS $j$ at the primary stage is
\begin{equation}\label{delta_Cjp}
\Delta C_{j}^{p}=\sum\limits_{(i,n)\in \mathcal{S}_{IN}^{o1}(j)}{{{a}_{ijn}}\Delta r_{{{i}_{p}}jn}^{p}}.
\end{equation}

Based on the above derivations, we can formulate the secondary user association and subchannel allocation subproblem as the following ILP problem
\begin{equation}\label{P5}
\begin{array}{l}
\mathop {\max }\limits_{{\mathcal A},{{\eta }_{R}},\{{{\eta }_{i}}\} } \ {\eta _R}\\
%-\rho \sum\limits_{i \in {\mathcal I}} \sum\limits_{j \in {\mathcal J}} {\sum\limits_{n \in {\mathcal N}} {a_{ijn}} }\\
s.t\\
\sum\nolimits_{i \in {\mathcal I}} {{a_{ijn}}}  \le 1,\forall j \in {\mathcal J},n \in {\mathcal N},\\
{\eta _i} = \sum\limits_{j \in {\mathcal J}} {\sum\limits_{n \in {\mathcal N}} {{a_{ijn}}r_{ijn}^s} } + \Delta R_i^p + R_i^p,\forall i \in {\mathcal I},\\
{\eta _i} \ge R_i^{\min },\forall i \in {\mathcal I},\\
{\eta _i} \ge {\eta _R},\forall i \in {\mathcal I},\\
\sum\limits_{i \in {\mathcal I}} {\sum\limits_{n \in {\mathcal N}} {{a_{ijn}}r_{ijn}^s } + \Delta C_{j}^{p} + C_j^p }  \le C_j^{\max },\forall j \in {\mathcal J},\\
{a_{ijn}} \in \{ 0,1\} ,\forall (i,j,n) \in {{\mathcal S}_2},\\
{a_{ijn}}{\rm{ = }}0,\forall (i,j,n) \notin {{\mathcal S}_2},
\end{array}
\end{equation}
where $r_{ijn}^s = \left\{ \begin{array}{*{35}{l}}
r_{ijn}^{u}, &\forall (i,j,n) \in {{\mathcal S}_u},\\
r_{ijn}^{o1}, &  \forall (i,j,n) \in {{\mathcal S}_{o1}},\\
r_{ijn}^{o2}, &  \forall (i,j,n) \in {{\mathcal S}_{o2}},\\
0, &  \forall (i,j,n) \notin {{\mathcal S}_2}.
\end{array} \right.$
%\begin{equation}\label{r_s}
%r_{ijn}^s = \left\{ \begin{array}{*{35}{l}}
%\!\! r_{ijn}^{u}, & \!\! \!\!  \forall (i,j,n) \in {{\mathcal S}_u}; & \!\! \!\!  r_{ijn}^{o1}, & \!\!  \!\! \forall (i,j,n) \in {{\mathcal S}_{o1}}\\
% \!\! r_{ijn}^{o2}, & \!\! \!\!  \forall (i,j,n) \in {{\mathcal S}_{o2}}; &  \!\! \!\! 0, &  \!\! \!\! \forall (i,j,n) \notin {{\mathcal S}_2}
%\end{array} \right.
%\end{equation}

Let $\{a_{ijn}^{s*}\}$ denote the solution of (\ref{P5}), and then the solution of (\ref{P3}) can be approximated as $\{a_{ijn}^{*}\}=\{a_{ijn}^{p*}+a_{ijn}^{s*}\}$.

\subsection{Power Control}
For any given user association and subchannel allocation $\mathcal{A}$, this subsection considers the subproblem of (\ref{P2}) of transmit power control.
%Particularly, we first determine the accepted slice requests according to the given slice request $\mathcal{A}$, and
Let ${{\mathcal{S}}_{a}}=\{(i,j,n)|{{a}_{ijn}}=1\}$ and ${{\mathcal{I}}_{a}}(j,n)=\{i|{{a}_{ijn}}=1\}$ represent the index set of the user association and subchannel allocation and the index set of UEs served by $\mathcal{SC}_{jn}$ respectively. Then we divide the user association and subchannel allocation into two categories. One category is that UE $i$ is a secondary UE on $\mathcal{S}{{\mathcal{C}}_{jn}}$, and let ${{\mathcal{S}}_{as}}=\{(i,j,n)\in {{\mathcal{S}}_{a}}|\sum\limits_{k\in \mathcal{I}}{{{a}_{kjn}}=2},{{h}_{ijn}}<\underset{k\in {{\mathcal{I}}_{_{a}}}(j,n)}{\mathop{\max }}\,{{h}_{kjn}}\}$ represent the index set of such user association and subchannel allocation. The other category is that UE $i$ is a primary UE on $\mathcal{SC}_{jn}$, and let ${{\mathcal{S}}_{ap}}={{\mathcal{S}}_{a}}\backslash {{\mathcal{S}}_{as}}$ represent the index set of such user association and subchannel allocation. Besides, let ${\mathcal S}_{JN}^{ap}(i)=\{(j,n)|(i,j,n)\in {{\mathcal{S}}_{ap}}\}$, ${\mathcal S}_{JN}^{as}(i)=\{(j,n)|(i,j,n)\in {{\mathcal{S}}_{as}}\}$, ${\mathcal S}_{IN}^{ap}(j)=\{(i,n)|(i,j,n)\in {{\mathcal{S}}_{ap}}\}$, and ${\mathcal S}_{IN}^{as}(j)=\{(i,n)|(i,j,n)\in {{\mathcal{S}}_{as}}\}$. Based on the above defined sets and by introducing auxiliary variables ${{\eta }_{R}}$, ${{\eta }_{P}}$ and $\{{{\eta }_{i}},\forall i\in \mathcal{I}\}$, the power control subproblem can be formulated as
\begin{subequations}\label{P6}
\begin{alignat}{2}
& \mathop {\max }\limits_{{\mathcal P},{{\eta }_{R}}, {{\eta }_{P}}, \{{{\eta }_{i}}\}, \{{p_j}\}} \quad {\eta _R} - {\eta _{EE}}{\eta _P}\\
& s.t \nonumber \\
& \sum\limits_{(j,n) \in {\mathcal S}_{JN}^{ap}(i)} {\!\!\!\!\!\!\!\!\!}{r_{ijn}^p({\mathcal P})}  + {\!\!\!\!\!\!\!\!\!}\sum\limits_{(j,n) \in {\mathcal S}_{JN}^{as}(i)} {\!\!\!\!\!\!\!\!\!}{r_{ijn}^s({\mathcal P})}  \ge {\eta _i},\forall i \in {\mathcal I},\\
& {\eta _i} \ge R_i^{\min },\forall i \in {\mathcal I},\\
& {\eta _i} \ge {\eta _R},\forall i \in {\mathcal I},\\
& \sum\limits_{(i,n) \in {\mathcal S}_{IN}^{ap}(j)} {\!\!\!\!\!\!\!\!\!}{r_{ijn}^p({\mathcal P})}  + {\!\!\!\!\!\!\!\!\!}\sum\limits_{(i,n) \in {\mathcal S}_{IN}^{as}(j)} {\!\!\!\!\!\!\!\!\!}{r_{ijn}^s({\mathcal P})}  \le C_j^{\max },\forall j \in {\mathcal J}, \displaybreak[0] \\
%& {p_{1,jn}} \le {p_{2,jn}},\forall (j,n) \in {\mathcal{S}_{c2}}\\
& {p_{j}} = \sum\limits_{n \in {\mathcal N}} {\left( {{p_{1,jn}} + {p_{2,jn}}} \right)},\forall j \in {\mathcal J},\\
& {p_j}  + p_j^c \le {{\eta }_{P}},\forall j \in {\mathcal J}, \\
& {p_j}  + p_j^c \le p_j^{\max },\forall j \in {\mathcal J}, \\
& {\text{(\ref{C3}), (\ref{C4}) and (\ref{C9}) are satisfied.}}
%& \sum\limits_{n \in {\mathcal N}} {\left( {{p_{1,jn}} + {p_{2,jn}}} \right)}  + p_j^c \le p_j^{\max },\forall j \in {\mathcal J}\\
%& {p_{1,jn}} \ge 0,{p_{2,jn}} \ge 0,\forall j \in {\mathcal J},n \in {\mathcal N}
\end{alignat}
\end{subequations}
where ${\eta_i}$ represents the achievable data transfer rate of UE $i$, ${\eta_R}$ represents the minimum achievable data transfer rate among all UEs and ${{\eta }_{P}}$ represents the maximum power consumption among all UBSs.
\begin{equation}\label{r_ijn1_P6}
\begin{array}{l}
r_{ijn}^p(\mathcal{P}) = \\
\frac{W}{N_s}{\log _2}\Big( {1 + \frac{{{p_{1,jn}}{h_{ijn}}}}{{\sum\limits_{k \ne j,k \in {\mathcal J}} {( {{p_{1,kn}} + {p_{2,kn}}} ){h_{ikn}}}  + p_n^Mh_{in}^M + \sigma _n^2}}} \Big),
\end{array}
\end{equation}
\begin{equation}\label{r_ijn2_P6}
\begin{array}{l}
r_{ijn}^s(\mathcal{P}) = \\
\frac{W}{N_s}{\log _2}\Big( {1 + \frac{{{p_{2,jn}}{h_{ijn}}}}{{{p_{1,jn}}{h_{ijn}} + {\!\!\!\!\!\!}\sum\limits_{k \ne j,k \in {\mathcal J}} {\!\!\!\!\!\!}{( {{p_{1,kn}} + {p_{2,kn}}} ){h_{ikn}}}  + p_n^Mh_{in}^M + \sigma _n^2}}} \Big).
\end{array}
\end{equation}

Note that in (\ref{P6}b) and (\ref{P6}e), $r_{ijn}^p({\mathcal P})$ and $r_{ijn}^s({\mathcal P})$ are neither convex nor concave with respect to $\mathcal{P}$. Thus, (\ref{P6}b) and (\ref{P6}e) are not convex constraints, and the problem (\ref{P6}) is a non-convex optimization problem.
%The following lemma show a method to tackle to the non-convexity of (\ref{P6}).
To solve this non-convex problem, we attempt to approximate the non-convex constraints as convex ones and then transform the non-convex problem into a convex one. To this aim, we resort to a successive convex approximation (SCA) approach \cite{boyd2004convex}. The SCA approach is an efficient technique to solve various types of non-convex optimization problems. The core idea of the SCA approach can be briefly described as approximating the original function as a more tractable function at a given point in each iteration. Specifically, the following lemma shows a method of tackling the non-convex constraints and approximating the non-convex (\ref{P6}) as a convex one. The approximated convex problem can be solved efficiently by existing optimization tools such as MOSEK \cite{mosek2019}.
\begin{lemma}\label{lem:2}
\rm{By exploring the SCA approach, (\ref{P6}) can be approximated into a convex optimization problem. Besides, the solution of the approximated problem is feasible for (\ref{P6}).}
\end{lemma}

\begin{proof}
Please refer to Appendix A.
%Provided in Appendix A in \cite{xi2020non}.
\end{proof}

\subsection{Iterative Association, Subchannel and Power Optimization}
Based on the above derivations, we propose an iterative association, subchannel and power optimization (IASPO) algorithm to solve (\ref{P2}), which is summarized in Algorithm \ref{alg1}. For convenience of description, we let ${\eta_R}={\mathop {\min }\nolimits_{i \in {\mathcal I}} {R_i}} $, ${{\eta_P}}= {\mathop {\max }\nolimits_{j \in {\mathcal J}} ({p_j}+ p_j^c)} $.
%For convenience of description, we let ${\eta_R}({\mathcal A},{\mathcal P})={\mathop {\min }\limits_{i \in {\mathcal I}} {R_i}({\mathcal A},{\mathcal P})} $, ${{\eta_P}}({\mathcal P})= {\mathop {\max }\limits_{j \in {\mathcal J}} {p_j}({\mathcal P})} $, and $\eta _{EE}({\mathcal A},{\mathcal P}) ={{{\eta_R}({\mathcal A},{\mathcal P})} \mathord{\left/ {\vphantom {{{\eta_R}({\mathcal A},{\mathcal P})} {{{\eta_P}}({\mathcal P})}}} \right. \kern-\nulldelimiterspace} {{{\eta_P}}({\mathcal P})}}$.
%Note that for the initial power ${{\mathcal P}^{(0)}}$, there may exist no feasible solutions of (\ref{P4}) and (\ref{P5}). Therefore, the IASPO algorithm does not consider the QoS requirement constraint and the capacity constraint when performing the subchannel allocation optimization in the first iteration (i.e., $r=0$).
Besides, the following lemma declares the convergence and complexity of the IASPO algorithm.
\begin{lemma}\label{lem:3}
\rm{The IASPO algorithm is convergent, and its complexity is $O( {{r}_{\max }}( 2({N_{u}}+1)^{N_d\cdot N_s} +(2N_d\cdot N_s)^{3.5} ))$ in the worst case.}
\end{lemma}
\begin{proof}
Please refer to Appendix B.
%Provided in Appendix B in \cite{xi2020non}.
\end{proof}

%It can prove that the IASPO algorithm is convergent. Due to the space limitation, we omit the proof of the convergence. The details are provided in the technical report \cite{}.
%It can prove that the convergence of the IASPO algorithm can always be guaranteed. Due to the space limitation, we omit the theoretical analysis of the convergence of the IASPO algorithm; yet, we will verify that the IASPO algorithm is convergent through the simulation.
%However, the IASPO algorithm can not guarantee optimality because an approximated association optimization method is developed and the exploration of the SCA approach and the iterative optimization method may lead to converging to a locally optimal solution \cite{boyd2004convex}.
%
%Based on the results of the above two sections, this section proposes an iterative algorithm, namely, iterative subchannel and power optimization (IASPO), for the problem (\ref{P2}), which is summarized in Algorithm \ref{alg1}. For convenience of description, we let ${\eta_R}({\mathcal A},{\mathcal P})=\left( {\mathop {\min }\limits_{i \in {\mathcal I}} {R_i}({\mathcal A},{\mathcal P})} \right)$, ${{\eta_P}}({\mathcal P})=\left( {\mathop {\max }\limits_{j \in {\mathcal J}} {p_i}({\mathcal P})} \right)$ and $\eta _{EE}({\mathcal A},{\mathcal P}) ={{{\eta_R}({\mathcal A},{\mathcal P})} \mathord{\left/
% {\vphantom {{{\eta_R}({\mathcal A},{\mathcal P})} {{{\eta_P}}({\mathcal P})}}} \right.
% \kern-\nulldelimiterspace} {{{\eta_P}}({\mathcal P})}}$.
%\begin{lemma}\label{lem:1}
%Algorithm \ref{alg1} is convergent.
%\end{lemma}
%

\begin{algorithm}
\caption{Iterative Association, Subchannel and Power Optimization}
\label{alg1}
\begin{algorithmic}[1]
\STATE \textbf{Initialize} ${{\mathcal P}^{(0)}}$, and let $r=0$.
\REPEAT
\STATE \textbf{Association and subchannel optimization:}
\STATE For given ${{\mathcal{P}}^{(r)}}$, obtain the solution by solving (\ref{P4}) and (\ref{P5}), and denote the solution by ${{\mathcal{A}}^{(r+1)}}$.
\IF {$ r > 0$ and ${\eta_R}({\mathcal A}^{(r+1)},{\mathcal P}^{(r)})<{\eta_R}({\mathcal A}^{(r)},{\mathcal P}^{(r)})$}
    \STATE Set ${\mathcal A}^{(r+1)}={\mathcal A}^{(r)}$.
\ENDIF
\STATE \textbf{Power optimization:}
\STATE For given ${{\mathcal{P}}^{(r)}}$ and ${{\mathcal{A}}^{(r+1)}}$, calculate $\eta_{EE} = \eta _{EE}({\mathcal A}^{(r+1)},{\mathcal P}^{(r)})$, obtain the solution by solving the approximated convex problem, and denote the solution by ${{\mathcal{P}}^{(r+1)}}$.
\STATE Update $r=r+1$.
\UNTIL Convergence or $r\ge {{r}_{\max }}$.
\end{algorithmic}
\end{algorithm}

\section{Simulation results}
%In this section, we conduct extensive simulations to validate the effectiveness of the proposed IASPO algorithms.
%Particularly, subsection A presents comparison algorithms and the parameter setting. Subsection B collects and analyzes the simulation results.

\subsection{Comparison Algorithms and Parameter Setting}

To our best knowledge, there are no existing works to be compared.
Therefore, to validate the effectiveness of the proposed IASPO algorithm, we compare the proposed algorithm with two benchmark algorithms:
1) Association and subchannel optimization-only (ASOO) algorithm: Allocate transmit power according to the initial power ${{\mathcal P}^{(0)}}$, and optimize user association and subchannel allocation by solving (\ref{P4}) and (\ref{P5}).
2) IASPO-FDMA algorithm: This algorithm is similar to the IASPO algorithm, except that it considers an FDMA-based communication scenario where each subchannel ${\mathcal{SC}_{jn}}$ can be allocated to at most one UE.
%\begin{itemize}
%\item Subchannel allocation optimization-only (ASOO) algorithm: Allocate transmit power according to the initial power
%    ${{\mathcal P}^{(0)}}$, and only allocate subchannel by solving (\ref{P4}) and (\ref{P5}).
%\item IASPO-FDMA algorithm: This algorithm is similar to the IASPO algorithm, except that it considers an OMA-based communication scenario where each slice can provide service to at most one UE.
%\end{itemize}

Set the size of the considered geographic area is a disc of radius $R_u=500$ m. The MBS is located at the center (0,0), and the UEs and UBSs are uniformly distributed in the annulus ($R_l$, $R_u$), where $R_l=250$ m. For each UE $i \in {\mathcal I}$, $R_{i}^{\min }$ is subject to a uniform distribution $U(R_{low}^{\min },R_{up}^{\min })$, and $R_{low}^{\min }=1$ Mb/s, $R_{up}^{\min }=2$ Mb/s. For each UBS $j \in {\mathcal J}$, $p_{j}^{c}=20$ dBm, $p_{j}^{\max }=24$ dBm, and $C_{j}^{\max }=100$ Mb/s.
%More simulation parameters are listed as below:
%$H$               =	100 m,
%$N$               =	4,	
%$W$               =	40 MHz,
%${{\alpha }_{1}}$ =	4.88,
%${{\alpha }_{2}}$ =	0.43,
%${p_n^M}$         =	251 mW (24 dBm),
%${\sigma _n^2}$   =	-85 dBm,
%${{f}_{c}}$       =	2.5 GHz,
%$c$               =	$3\times {{10}^{8}}$ m/s,
%$d_0$	            = 1,	
%$g_{ijn}^{Tx}$ = $g_{ijn}^{Rx}$ = $g_{in}^{MTx}$ = $g_{in}^{MRx}$ = 1,
%$\eta _{LoS}^{dB}$	= 0.1,
%$\eta _{NLoS}^{dB}$	= 21,
%$\eta$	            = 3,
%$r_{max}$	        = 1000.	
More simulation parameters are listed in Table \ref{parameter_setting}.
\begin{table}[!t]
\renewcommand{\arraystretch}{1.3}
%% increase table row spacing, adjust to taste
%\renewcommand{\arraystretch}{1.3}
% if using array.sty, it might be a good idea to tweak the value of
% \extrarowheight as needed to properly center the text within the cells
\caption{System parameters}
\label{parameter_setting}
\centering
%% Some packages, such as MDW tools, offer better commands for making tables
%% than the plain LaTeX2e tabular which is used here.
\begin{tabular}{|c|c|c|c|c|c|}
\hline
Parameters        &	Value	                  & Parameters	        & Value	\\\hline
$H$               &	100 m	                  & $d_0$	            & 1	    \\\hline
$N_s$               &	4	                      & $g_{ijn}^{Tx}$	    & 1	    \\\hline
$W$               &	40 MHz	                  & $g_{ijn}^{Rx}$	    & 1	    \\\hline
${{\alpha }_{1}}$ &	4.88	                  & $g_{in}^{MTx}$	    & 1	    \\\hline
${{\alpha }_{2}}$ &	0.43	                  & $g_{in}^{MRx}$	    & 1	    \\\hline
${p_n^M}$         &	24 dBm	                  & $\eta _{LoS}^{dB}$	& 0.1   \\\hline
${\sigma _n^2}$   &	-85 dBm	                  & $\eta _{NLoS}^{dB}$	& 21	\\\hline
${{f}_{c}}$       &	2.5 GHz                   & $\eta$	            & 3	    \\\hline
$c$               &	$3\times {{10}^{8}}$ m/s  & $r_{max}$	        & 1000	\\\hline
\end{tabular}
\end{table}

\subsection{Performance Evaluation}
All comparison algorithms need to initialize ${{\mathcal P}^{(0)}}$. For the algorithms except the IASPO-FDMA algorithm, we initialize ${{\mathcal P}^{(0)}}$ to $p_{1,jn}^{(0)} = p_{2,jn}^{(0)} = \frac{p_{j}^{\max }-p_{j}^{c}}{4N_s}$ for all $j \in {\mathcal J}$, $n \in {\mathcal N}$. For the IASPO-FDMA algorithm, we initialize ${{\mathcal P}^{(0)}}$ to $p_{jn}^{(0)}=\frac{p_{j}^{\max }-p_{j}^{c}}{2N_s}$ for all $j \in {\mathcal J}$, $n \in {\mathcal N}$.

We perform all comparison algorithms on one hundred randomly generated data sets in the simulation, and the final result is the average of the one hundred results.

We first study the convergence of the proposed IASPO algorithm. Fig. \ref{objective_vs_iteration} illustrates the convergence behaviour of the energy efficiency $f_{EE}$ of the IASPO algorithm. We can observe that $f_{EE}$ increases monotonously with the increase of the iteration index and quickly converges to a certain value.
\begin{figure}[!t]
\centering
\includegraphics[width=3in]{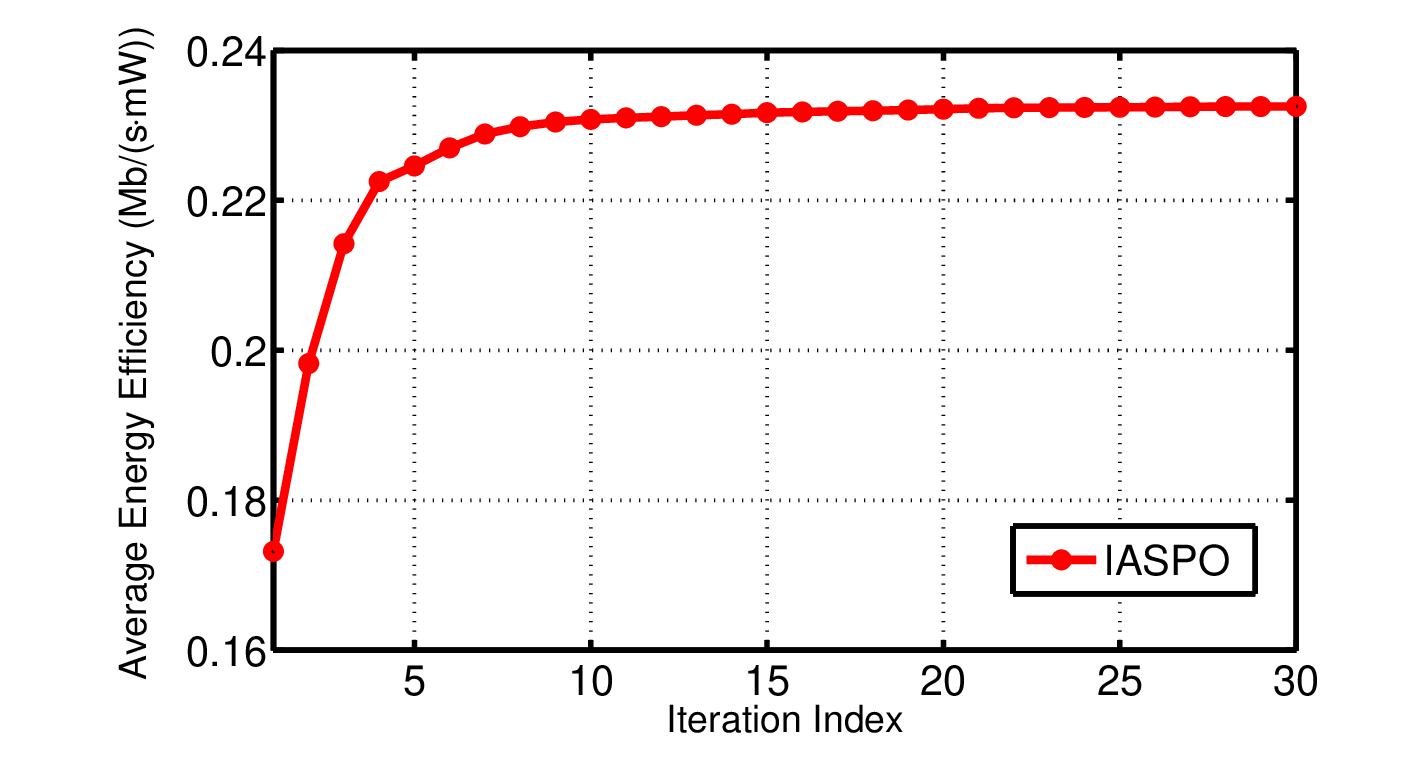}
% where an .eps filename suffix will be assumed under latex,
% and a .pdf suffix will be assumed for pdflatex; or what has been declared
% via \DeclareGraphicsExtensions.
\caption{Average energy efficiency vs. iteration index when $N_{u}=10$, $N_{d}=4$.}
\label{objective_vs_iteration}
\end{figure}

Then we consider the effect of the number of UEs $N_{u}$ and the number of UBSs $N_{d}$ on the energy efficiency $f_{EE}$ for all comparison algorithms. Fig. \ref{efficiency_vs_I} illustrates the energy efficiency vs. the number of UEs and Fig. \ref{efficiency_vs_J} illustrates the energy efficiency vs. the number of UBSs. From Figs. \ref{efficiency_vs_I}, \ref{efficiency_vs_J}, we can observe that:
\begin{itemize}
\item  The IASPO algorithm can achieve the highest energy efficiency compared with the other two algorithms except when $N_{u}=6$ and $N_{d}=4$. Given the number of UAVs (e.g., $N_{d}=4$), when the number of UEs is greater than six, the IASPO algorithm outperforms the IASPO-FDMA algorithm. When the number of UEs is six, the performance of the IASPO-FDMA algorithm is better than that of the the IASPO algorithm. This is because NOMA affects the resource allocation in our model in two aspects: 1) \emph{pros}: improve the spectrum efficiency of the network. 2) \emph{cons}: the association and subchannel optimization can only obtain an approximate solution due to the increased computational complexity. When the number of UEs is great, the exploitation of the NOMA technique improves the spectrum efficiency. However, when the number of UEs is small, the loss of exploiting the NOMA technique is greater than the benefit.
\item  The energy efficiency of the ASOO and IASPO-FDMA algorithms generally decreases with the increase of the number of UEs, while the energy efficiency of the IASPO algorithm is relatively robust to the increase of the number of UEs.
\item  The energy efficiency of all comparison algorithms generally increases with the increasing number UBS. However, the deployment of more UBSs means the consumption of more UAV resources.
    In summary, the above results indicate that our proposed IASPO algorithm can improve the energy efficiency, especially in the scenario where UAV resources are relatively scarce, that is, there are many UEs or few UBSs.
\end{itemize}

\begin{figure}[!t]
\centering
\includegraphics[width=3in]{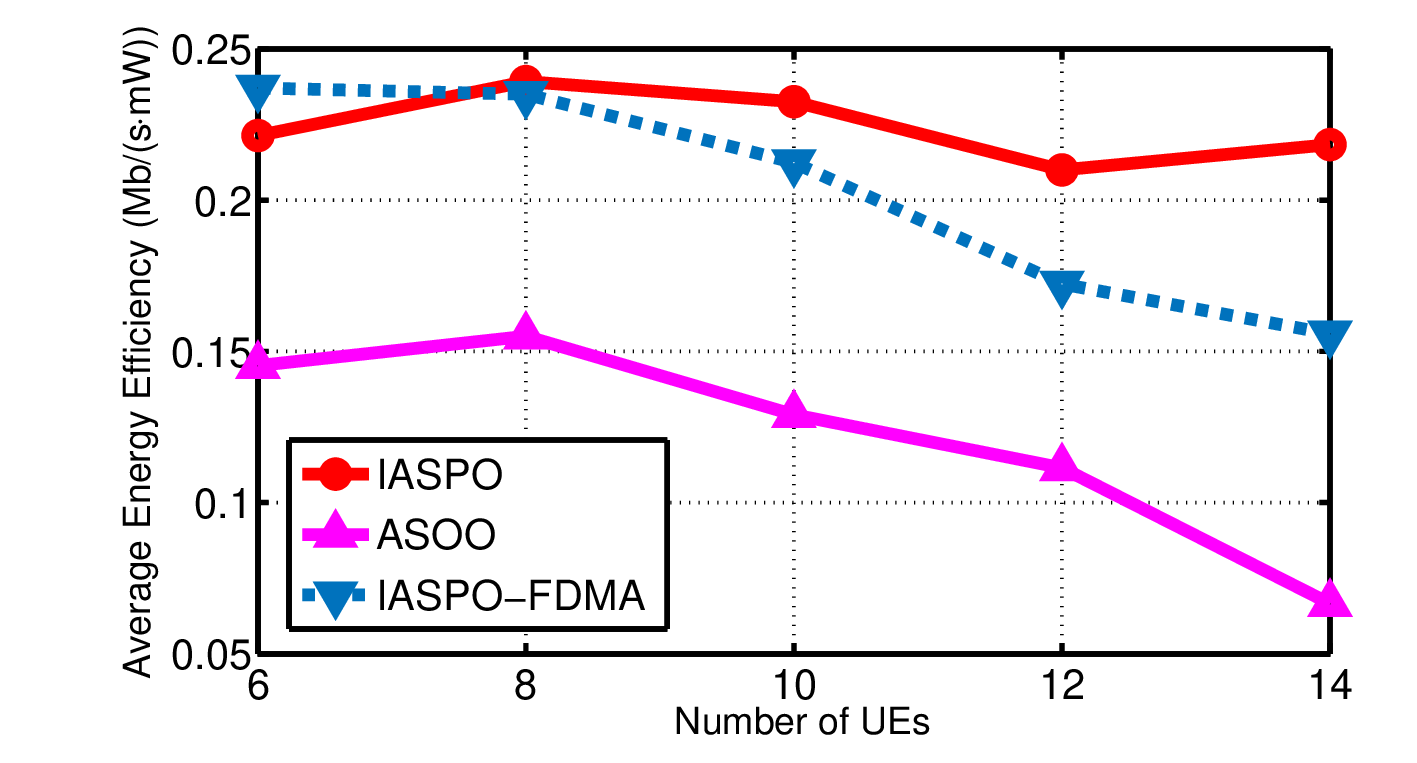}
% where an .eps filename suffix will be assumed under latex,
% and a .pdf suffix will be assumed for pdflatex; or what has been declared
% via \DeclareGraphicsExtensions.
\caption{Average energy efficiency vs. the number of UEs when $N_{d}=4$.}
\label{efficiency_vs_I}
\end{figure}
\begin{figure}[!t]
\centering
\includegraphics[width=3in]{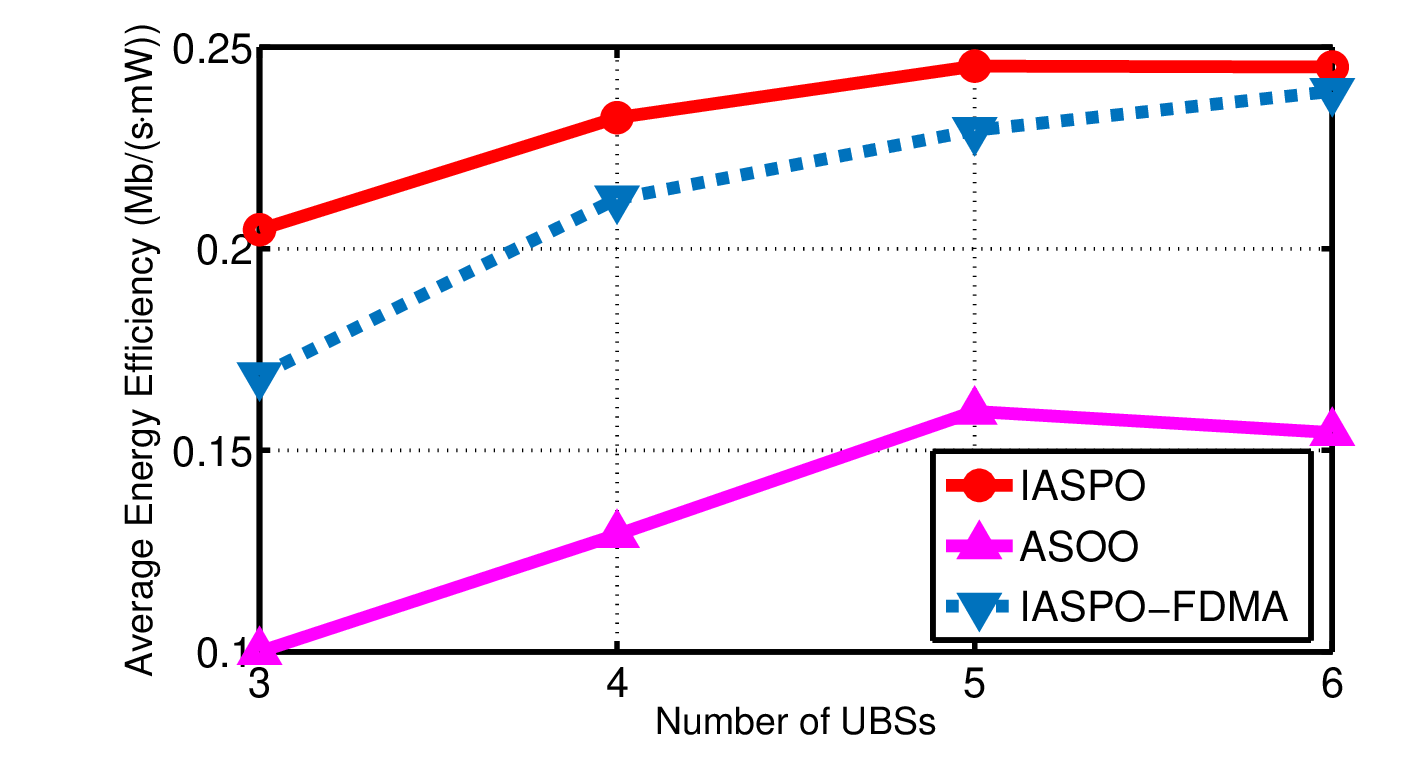}
% where an .eps filename suffix will be assumed under latex,
% and a .pdf suffix will be assumed for pdflatex; or what has been declared
% via \DeclareGraphicsExtensions.
\caption{Average energy efficiency vs. the number of UBSs when $N_{u}=10$.}
\label{efficiency_vs_J}
\end{figure}

\section{Conclusion}
This paper investigated the non-orthogonal resource allocation problem in a multi-UAV-aided network for providing eMBB services and formulated a joint non-orthogonal user association, subchannel allocation and power control problem to maximize the network energy efficiency. To alleviate this problem, we first decomposed it into two separated subproblems, namely, a user association and subchannel allocation subproblem and a power control subproblem.
We then designed a two-stage approximation strategy to solve the user association and subchannel allocation subproblem and exploited an SCA approach to approximate the power control subproblem.
Based on the above results, we then developed an iterative algorithm with provable convergence to solve the joint optimization problem. Simulation results verified that our proposed algorithm could improve the energy efficiency compared with several benchmark algorithms.
This paper assumes that each subchannel ${\mathcal{SC}_{jn}}$ can be assigned to at most two users, developing some low-complexity algorithms considering more-user NOMA may be a topic worthy of research in the near future.

%\section*{Acknowledgment}
%This work was supported by the National Natural Science Foundation of China under Grant 91738301 and Grant 61827901.


\appendix

\subsection{Proof of Lemma 2}
\begin{proof}
%The core idea of the SCA approach can be briefly described as approximating the original function as a more tractable function at a given point in each iteration.
Let ${{\mathcal P}^{(r)}}=\{p_{1,jn}^{(r)},p_{2,jn}^{(r)}\}$ denote the given transmit power point in the $(r+1)$-th iteration ($r\ge0$). Next, we discuss how to transform (\ref{P6}) into a convex optimization problem via the SCA approach in detail. Note that we need to approximate the left-hand side of (\ref{P6}b) as a concave function and the left-hand side of (\ref{P6}e) as a convex function.

First, we study the approximation of $r_{ijn}^p({\mathcal P})$ and $r_{ijn}^s({\mathcal P})$. For $r_{ijn}^p({\mathcal P})$, it can be written as a difference of two concave functions with respect to $\mathcal{P}$, i.e.,
\begin{equation}\label{r1}
r_{ijn}^p({\mathcal P}) = \mathord{\buildrel{\lower3pt\hbox{$\scriptscriptstyle\frown$}}
\over r} _{ijn}^p({\mathcal P}) - \mathord{\buildrel{\lower3pt\hbox{$\scriptscriptstyle\smile$}}
\over r} _{ijn}^p({\mathcal P}),
\end{equation}
where
\begin{equation}\label{r11}
\begin{array}{l}
\mathord{\buildrel{\lower3pt\hbox{$\scriptscriptstyle\frown$}}
\over r} _{ijn}^p({\mathcal P}) = \\
\frac{W}{N_s}{\log _2}\big( {{p_{1,jn}}{h_{ijn}} + {\!\!\!\!\!\!\!}\sum\limits_{k \ne j,k \in {\mathcal J}}{\!\!\!\!\!\!\!} {( {{p_{1,kn}} + {p_{2,kn}}} ){h_{ikn}}}  + p_n^Mh_{in}^M + \sigma _n^2} \big),
\end{array}
\end{equation}
\begin{equation}\label{r12}
\begin{array}{l}
\mathord{\buildrel{\lower3pt\hbox{$\scriptscriptstyle\smile$}}
\over r} _{ijn}^p({\mathcal P}) =
\frac{W}{N_s}{\log _2}\big( {{\!\!\!\!\!\!}\sum\limits_{k \ne j,k \in {\mathcal J}}{\!\!\!\!\!\!} {( {{p_{1,kn}} + {p_{2,kn}}} ){h_{ikn}}}  + p_n^Mh_{in}^M + \sigma _n^2} \big).
\end{array}
\end{equation}

It can be proved that any concave function is globally upper-bounded by its first-order Taylor expansion at any point \cite{boyd2004convex}. Therefore, we have the following upper bounds of $\overset{\scriptscriptstyle\frown}{r}_{ijn}^{p}(\mathcal{P})$ and $\overset{\scriptscriptstyle\smile}{r}_{ijn}^{p}(\mathcal{P})$ at the given transmit power point ${{\mathcal{P}}^{(r)}}$
\begin{equation}\label{r11up}
\begin{array}{l}
\mathord{\buildrel{\lower3pt\hbox{$\scriptscriptstyle\frown$}}
\over r} _{ijn}^p({\mathcal P}) \le
B_{ijn}^{(r)} + D_{ijn}^{(r)}\Big({h_{ijn}} \big( {{p_{1,jn}} - p{_{1,jn}^{(r)}}} \big) + \\
\sum\limits_{k \ne j,k \in {\mathcal J}} {{h_{ikn}}\big({{p_{1,kn}} + {p_{2,kn}} - p{_{1,kn}^{(r)}} - p{_{2,kn}^{(r)}}}\big)} \Big)\\
 = \mathord{\buildrel{\lower3pt\hbox{$\scriptscriptstyle\frown$}}
\over r} _{ijn}^{p,t(r)}({\mathcal P}),
\end{array}
\end{equation}
\begin{equation}\label{r12up}
\begin{array}{l}
\mathord{\buildrel{\lower3pt\hbox{$\scriptscriptstyle\smile$}}
\over r} _{ijn}^p({\mathcal P}) \le E_{ijn}^{(r)} + \\
F_{ijn}^{(r)}\Big( {\sum\limits_{k \ne j,k \in {\mathcal J}} {{h_{ikn}}\big( {{p_{1,kn}} + {p_{2,kn}} - p{_{1,kn}^{(r)}} - p{_{2,kn}^{(r)}}} \big)} } \Big)\\
 = \mathord{\buildrel{\lower3pt\hbox{$\scriptscriptstyle\smile$}}
\over r} _{ijn}^{p,t(r)}({\mathcal P}),
\end{array}
\end{equation}
where
\begin{equation}\label{Bijn}
\begin{array}{l}
B_{ijn}^{(r)} = \frac{W}{N_s}{\log _2}\Big(p_{1,jn}^{(r)}{h_{ijn}} + \\
\sum\limits_{k \ne j,k \in {\mathcal J}} {\big( {p_{1,kn}^{(r)} + p_{2,kn}^{(r)}} \big){h_{ikn}}}  + p_n^Mh_{in}^M + \sigma _n^2 \Big),
\end{array}
\end{equation}
\begin{equation}\label{Dijn}
\begin{array}{l}
D_{ijn}^{(r)} =
\frac{{{{\log }_2}(e)W/N_s}}{{p_{1,jn}^{(r)}{h_{ijn}} + {\!\!\!\!\!\!}\sum\limits_{k \ne j,k \in {\mathcal J}} {\!\!\!\!\!\!}{({p_{1,kn}^{(r)} + p_{2,kn}^{(r)}}){h_{ikn}}}  + p_n^Mh_{in}^M + \sigma _n^2}},
\end{array}
\end{equation}
\begin{equation}\label{Eijn}
\begin{array}{l}
E_{ijn}^{(r)} = \\
\frac{W}{N_s}{\log _2}\Big( {\sum\limits_{k \ne j,k \in {\mathcal J}} {\!\!\!\!\!\!} {\big( {p_{1,kn}^{(r)} + p_{2,kn}^{(r)}} \big){h_{ikn}}}  + p_n^Mh_{in}^M + \sigma _n^2} \Big),
\end{array}
\end{equation}
\begin{equation}\label{Fijn}
F_{ijn}^{(r)} = \frac{{{{\log }_2}(e)W/N_s}}{{\sum\limits_{k \ne j,k \in {\mathcal J}} {\big( {p_{1,kn}^{(r)} + p_{2,kn}^{(r)}} \big){h_{ikn}}}  + p_n^Mh_{in}^M + \sigma _n^2}}.
\end{equation}

Similarly, $r_{ijn}^s(\mathcal{P})$ can be written as a difference of two concave functions with respect to $\mathcal{P}$, i.e.,

\begin{equation}\label{r2}
r_{ijn}^s({\mathcal P}) = \mathord{\buildrel{\lower3pt\hbox{$\scriptscriptstyle\frown$}}
\over r} _{ijn}^s({\mathcal P}) - \mathord{\buildrel{\lower3pt\hbox{$\scriptscriptstyle\smile$}}
\over r} _{ijn}^s({\mathcal P}),
\end{equation}
where
\begin{equation}\label{r21}
\begin{array}{l}
\mathord{\buildrel{\lower3pt\hbox{$\scriptscriptstyle\frown$}}
\over r} _{ijn}^s({\mathcal P}) =
\frac{W}{N_s}{\log _2}\big( {\sum\limits_{k \in {\mathcal J}} {( {{p_{1,kn}} + {p_{2,kn}}} ){h_{ikn}}}  + p_n^Mh_{in}^M + \sigma _n^2} \big),
\end{array}
\end{equation}
\begin{equation}\label{r22}
\begin{array}{l}
\mathord{\buildrel{\lower3pt\hbox{$\scriptscriptstyle\smile$}}
\over r} _{ijn}^s(P) = \\
\frac{W}{N_s}{\log _2}\big( {{p_{1,jn}}{h_{ijn}} + {\!\!\!\!\!\!}\sum\limits_{k \ne j,k \in {\mathcal J}} {\!\!\!\!\!\!}{({{p_{1,kn}} + {p_{2,kn}}}){h_{ikn}}}  + p_n^Mh_{in}^M + \sigma _n^2} \big).
\end{array}
\end{equation}

It can be observed that $\overset{\scriptscriptstyle\smile}{r}_{ijn}^{s}(\mathcal{P})$ and $\overset{\scriptscriptstyle\frown}{r}_{ijn}^{p}(\mathcal{P})$ have an identical form. Thus, the upper bound of $\overset{\scriptscriptstyle\smile}{r}_{ijn}^{s}(\mathcal{P})$ at ${{\mathcal{P}}^{(r)}}$ can be expressed as
\begin{equation}\label{r22up}
\begin{array}{l}
\mathord{\buildrel{\lower3pt\hbox{$\scriptscriptstyle\smile$}}
\over r} _{ijn}^s({\mathcal P}) \le
B_{ijn}^{(r)} + D_{ijn}^{(r)}\Big({h_{ijn}} \big( {{p_{1,jn}} - p{_{1,jn}^{(r)}}} \big) +\\
\sum\limits_{k \ne j,k \in {\mathcal J}} {{h_{ikn}}\big({{p_{1,kn}} + {p_{2,kn}} - p{_{1,kn}^{(r)}} - p{_{2,kn}^{(r)}}}\big)} \Big)\\
=\mathord{\buildrel{\lower3pt\hbox{$\scriptscriptstyle\smile$}}
\over r} _{ijn}^{s,t(r)}({\mathcal P}).
\end{array}
\end{equation}

For $\overset{\scriptscriptstyle\frown}{r}_{ijn}^{s}(\mathcal{P})$, by leveraging the first-order Taylor expansion, we have the following upper bound of $\overset{\scriptscriptstyle\frown}{r}_{ijn}^{s}(\mathcal{P})$ at ${{\mathcal{P}}^{(r)}}$
\begin{equation}\label{r21up}
\begin{array}{l}
\mathord{\buildrel{\lower3pt\hbox{$\scriptscriptstyle\frown$}}
\over r} _{ijn}^s({\mathcal P}) \le G_{ijn}^{(r)} + \\
H_{ijn}^{(r)}\Big( {\sum\limits_{k \in {\mathcal J}} {{h_{ikn}}\big( {{p_{1,kn}} + {p_{2,kn}}  - p_{1,kn}^{(r)} - p_{2,kn}^{(r)}} \big)} } \Big)\\
 = \mathord{\buildrel{\lower3pt\hbox{$\scriptscriptstyle\frown$}}
\over r} _{ijn}^{s,t(r)}({\mathcal P}),
\end{array}
\end{equation}
where
\begin{equation}\label{Gijn}
\begin{array}{l}
G_{ijn}^{(r)} =
\frac{W}{N_s}{\log _2}\Big( {\sum\limits_{k \in {\mathcal J}} {\big( {p_{1,kn}^{(r)} + p_{2,kn}^{(r)}} \big){h_{ikn}}}  + p_n^Mh_{in}^M + \sigma _n^2} \Big),
\end{array}
\end{equation}
\begin{equation}\label{Hijn}
H_{ijn}^{(r)} = \frac{{{{\log }_2}(e)W/N_s}}{{\sum\limits_{k \in {\mathcal J}} {\big( {p_{1,kn}^{(r)} + p_{2,kn}^{(r)}} \big){h_{ikn}}}  + p_n^Mh_{in}^M + \sigma _n^2}}.
\end{equation}

It can be observed that the upper bounds $\overset{\scriptscriptstyle\frown}{r}_{ijn}^{p,t(r)}(\mathcal{P})$, $\overset{\scriptscriptstyle\smile}{r}_{ijn}^{p,t(r)}(\mathcal{P})$, $\overset{\scriptscriptstyle\frown}{r}_{ijn}^{s,t(r)}(\mathcal{P})$, and $\overset{\scriptscriptstyle\smile}{r}_{ijn}^{s,t(r)}(\mathcal{P})$ are linear functions with respect to $\mathcal{P}$.

Next, we study the approximation of the constraints (\ref{P6}b) and (\ref{P6}e).
By substituting (\ref{r12up}) into (\ref{r1}) and substituting (\ref{r22up}) into (\ref{r2}), for all $i\in \mathcal{I}$, we can obtain the lower bound of the left-hand side of the constraint (\ref{P6}b) as
\begin{equation}\label{LHS_P6b}
\begin{array}{l}
\sum\limits_{(j,n) \in {\mathcal S}_{JN}^{ap}(i)} {r_{ijn}^p({\mathcal P})}  + \sum\limits_{(j,n) \in {\mathcal S}_{JN}^{as}(i)} {r_{ijn}^s({\mathcal P})}  \\
\ge \sum\limits_{(j,n) \in {\mathcal S}_{JN}^{ap}(i)} {\left( {\mathord{\buildrel{\lower3pt\hbox{$\scriptscriptstyle\frown$}}
\over r} _{ijn}^p({\mathcal P}) - \mathord{\buildrel{\lower3pt\hbox{$\scriptscriptstyle\smile$}}
\over r} _{ijn}^{p,t(r)}({\mathcal P})} \right)} \\
 + \sum\limits_{(j,n) \in {\mathcal S}_{JN}^{as}(i)} {\left( {\mathord{\buildrel{\lower3pt\hbox{$\scriptscriptstyle\frown$}}
\over r} _{ijn}^s({\mathcal P}) - \mathord{\buildrel{\lower3pt\hbox{$\scriptscriptstyle\smile$}}
\over r} _{ijn}^{s,t(r)}({\mathcal P})} \right)}.
\end{array}
\end{equation}

Similarly, by substituting (\ref{r11up}) into (\ref{r1}) and substituting (\ref{r21up}) into (\ref{r2}), for all $j\in \mathcal{J}$, we obtain the upper bound of the left-hand side of the constraint (\ref{P6}e) as
\begin{equation}\label{LHS_P6e}
\begin{array}{l}
\sum\limits_{(i,n) \in {\mathcal S}_{IN}^{ap}(j)} {r_{ijn}^p({\mathcal P})}  + \sum\limits_{(i,n) \in {\mathcal S}_{IN}^{as}(j)} {r_{ijn}^s({\mathcal P})}  \\
\le \sum\limits_{(i,n) \in {\mathcal S}_{IN}^{ap}(j)} {\left( {\mathord{\buildrel{\lower3pt\hbox{$\scriptscriptstyle\frown$}}
\over r} _{ijn}^{p,t(r)}({\mathcal P}) - \mathord{\buildrel{\lower3pt\hbox{$\scriptscriptstyle\smile$}}
\over r} _{ijn}^p({\mathcal P})} \right)} \\
 + \sum\limits_{(i,n) \in {\mathcal S}_{IN}^{as}(j)} {\left( {\mathord{\buildrel{\lower3pt\hbox{$\scriptscriptstyle\frown$}}
\over r} _{ijn}^{s,t(r)}({\mathcal P}) - \mathord{\buildrel{\lower3pt\hbox{$\scriptscriptstyle\smile$}}
\over r} _{ijn}^s({\mathcal P})} \right)}.
\end{array}
\end{equation}

Therefore, with any given transmit power point ${{\mathcal{P}}^{(r)}}=\{p_{1,jn}^{(r)},p_{2,jn}^{(r)}\}$, the problem (\ref{P6}) can be approximated as the following form by referring to (\ref{LHS_P6b}) and (\ref{LHS_P6e})
\begin{subequations}\label{P7}
\begin{alignat}{2}
& \mathop {\max }\limits_{{\mathcal P},{{\eta }_{R}}, {{\eta }_{P}}, \{{{\eta }_{i}}\}, \{{p_j}\}} {\mkern 1mu} \quad {\eta _R} - {\eta _{EE}}{\eta _P}\\
& s.t \nonumber \\
&\sum\limits_{(j,n) \in {\mathcal S}_{JN}^{ap}(i)} {\!\!\!\!\!\!}{\left( {\mathord{\buildrel{\lower3pt\hbox{$\scriptscriptstyle\frown$}}
\over r} _{ijn}^p({\mathcal P}) - \mathord{\buildrel{\lower3pt\hbox{$\scriptscriptstyle\smile$}}
\over r} _{ijn}^{p,t(r)}({\mathcal P})} \right)}  + \nonumber\\
&\sum\limits_{(j,n) \in {\mathcal S}_{JN}^{as}(i)} {\!\!\!\!\!\!}{\left( {\mathord{\buildrel{\lower3pt\hbox{$\scriptscriptstyle\frown$}}
\over r} _{ijn}^s({\mathcal P}) - \mathord{\buildrel{\lower3pt\hbox{$\scriptscriptstyle\smile$}}
\over r} _{ijn}^{s,t(r)}({\mathcal P})} \right)}  \ge {\eta _i},\forall i \in {\mathcal I},  \\
%& {\eta _i} \ge R_i^{\min },\forall i \in {\mathcal I}\\
%& {\eta _i} \ge {\eta _R},\forall i \in {\mathcal I}\\
&\sum\limits_{(i,n) \in {\mathcal S}_{IN}^{ap}(j)} {\!\!\!\!\!\!}{\left( {\mathord{\buildrel{\lower3pt\hbox{$\scriptscriptstyle\frown$}}
\over r} _{ijn}^{p,t(r)}({\mathcal P}) - \mathord{\buildrel{\lower3pt\hbox{$\scriptscriptstyle\smile$}}
\over r} _{ijn}^p({\mathcal P})} \right)}  +  \nonumber\\
&\sum\limits_{(i,n) \in {\mathcal S}_{IN}^{as}(j)} {\!\!\!\!\!\!}{\left( {\mathord{\buildrel{\lower3pt\hbox{$\scriptscriptstyle\frown$}}
\over r} _{ijn}^{s,t(r)}({\mathcal P}) - \mathord{\buildrel{\lower3pt\hbox{$\scriptscriptstyle\smile$}}
\over r} _{ijn}^s({\mathcal P})} \right)} \le C_j^{\max },\forall j \in {\mathcal J}, \\
%& {p_{1,jn}} \le {p_{2,jn}},\forall (j,n) \in {\mathcal{S}_{c2}}\\
%& \sum\limits_{n \in {\mathcal N}} {\left( {{p_{1,jn}} + {p_{2,jn}}} \right)}  + p_j^c \le {\eta _P},\forall j \in {\mathcal J}\\
& {\text{ (\ref{P6}c), (\ref{P6}d), (\ref{P6}f), (\ref{P6}g), (\ref{P6}h) and (\ref{P6}i) are satisfied.}}\nonumber
%& \sum\limits_{n \in {\mathcal N}} {\left( {{p_{1,jn}} + {p_{2,jn}}} \right)}  + p_j^c \le p_j^{\max },\forall j \in {\mathcal J}\\
%& {p_{1,jn}} \ge 0,{p_{2,jn}} \ge 0,\forall j \in {\mathcal J},n \in {\mathcal N}
\end{alignat}
\end{subequations}

Since the left-hand sides of the constraints (\ref{P7}b) and (\ref{P7}c) are concave and convex with respect to $\mathcal{P}$ respectively, (\ref{P7}b) and (\ref{P7}c) are convex constraints. Therefore, the problem (\ref{P7}) is a convex optimization problem.

Note that the inequalities (\ref{LHS_P6b}) and (\ref{LHS_P6e}) indicate that any feasible solution of the problem (\ref{P7}) is also feasible for the problem (\ref{P6}), but the reverse is not true in general. Therefore, the optimal objective value obtained by solving (\ref{P7}) is the lower bound of that of (\ref{P6}).
\end{proof}

\subsection{Proof of Lemma 3}
\begin{proof}
%For convenience of description, we let ${\eta_R({\mathcal A},{\mathcal P})}={\mathop {\min }\nolimits_{i \in {\mathcal I}} {R_i}({\mathcal A},{\mathcal P})} $, ${{\eta_P}({\mathcal P})}= {\mathop {\max }\nolimits_{j \in {\mathcal J}} ({p_j}({\mathcal P})+{p_j^c})} $.
In the $r$-th iteration ($r\ge1$), the obtained $\eta _{EE}({\mathcal A}^{(r)},{\mathcal P}^{(r)})$ can be expressed as
\begin{equation}\label{eta_1}
{\eta _{EE}}({{\mathcal A}^{(r)}},{{\mathcal P}^{(r)}}) = \frac{{{\eta_R}({{\mathcal A}^{(r)}},{{\mathcal P}^{(r)}})}}{{{{\eta_P}}({{\mathcal P}^{(r)}})}}.
\end{equation}

Then in the $(r+1)$-th iteration, after performing the association and subchannel optimization, we can obtain
\begin{equation}\label{eta_2}
{{\eta_R}({{\mathcal A}^{(r+1)}},{{\mathcal P}^{(r)}})} \ge {{\eta_R}({{\mathcal A}^{(r)}},{{\mathcal P}^{(r)}})}.
\end{equation}

According to (\ref{eta_1}) and (\ref{eta_2}), $\eta _{EE}({\mathcal A}^{(r+1)},{\mathcal P}^{(r)})$ satisfies
\begin{equation}\label{eta_3}
\begin{array}{l}
{\eta _{EE}}({\mathcal{A}^{(r+1)}},{\mathcal{P}^{(r)}}) = \frac{{{\eta_R}({\mathcal{A}^{(r{\rm{ + }}1)}},{\mathcal{P}^{(r)}})}}{{{{\eta_P}}({\mathcal{P}^{(r)}})}}\\
 \ge \frac{{{\eta_R}({\mathcal{A}^{(r)}},{\mathcal{P}^{(r)}})}}{{{{\eta_P}}({\mathcal{P}^{(r)}})}}{\rm{ = }}{\eta _{EE}}({\mathcal{A}^{(r)}},{\mathcal{P}^{(r)}}).
\end{array}
\end{equation}

After performing the power optimization, we can obtain
\begin{equation}\label{eta_4}
\begin{array}{l}
{\eta_R}({{\mathcal A}^{(r + 1)}},{{\mathcal P}^{(r + 1)}}) - {\eta _{EE}}({{\mathcal A}^{(r + 1)}},{{\mathcal P}^{(r)}}){{\eta_P}}({{\mathcal P}^{(r + 1)}})\\
 \ge {\eta_R}({{\mathcal A}^{(r + 1)}},{{\mathcal P}^{(r)}}) - {\eta _{EE}}({{\mathcal A}^{(r + 1)}},{{\mathcal P}^{(r)}}){{\eta_P}}({{\mathcal P}^{(r)}}) = 0.
\end{array}
\end{equation}

Thus we can obtain
\begin{equation}\label{eta_5}
{\eta_R}({{\mathcal A}^{(r + 1)}},{{\mathcal P}^{(r + 1)}}) \ge {\eta _{EE}}({{\mathcal A}^{(r + 1)}},{{\mathcal P}^{(r)}}){{\eta_P}}({{\mathcal P}^{(r + 1)}}),
\end{equation}
\begin{equation}\label{eta_6}
\begin{array}{l}
{\eta _{EE}}({{\mathcal A}^{(r + 1)}},{{\mathcal P}^{(r + 1)}}) = \frac{{{\eta_R}({{\mathcal A}^{(r + 1)}},{{\mathcal P}^{(r + 1)}})}}{{{{\eta_P}}({{\mathcal P}^{(r + 1)}})}}\\
 \ge {\eta _{EE}}({{\mathcal A}^{(r + 1)}},{{\mathcal P}^{(r)}}).
\end{array}
\end{equation}

According to (\ref{eta_3}) and (\ref{eta_6}), we can obtain
\begin{equation}\label{eta_7}
{\eta _{EE}}({{\mathcal A}^{(r + 1)}},{{\mathcal P}^{(r + 1)}}) \ge {\eta _{EE}}({{\mathcal A}^{(r)}},{{\mathcal P}^{(r)}}).
\end{equation}
and thus the convergence of the IASPO algorithm is proved.

The complexity of the IASPO algorithm is dominated by that of solving (\ref{P4}), (\ref{P5}) and (\ref{P7}). The complexities of solving the ILP problems (\ref{P4}), (\ref{P5}) are both $O( ({N_{u}}+1)^{N_d\cdot N_s} )$, and the complexity of solving the convex problem (\ref{P7}) is $O( (2N_d\cdot N_s)^{3.5} )$. Moreover, since (\ref{P4}), (\ref{P5}) and (\ref{P7}) need to be iteratively solved until the IASPO algorithm converges or reaches the maximum number of iterations ${{r}_{\max }}$, the complexity of the IASPO algorithm is $O( {{r}_{\max }}( 2({N_{u}}+1)^{N_d\cdot N_s} + (2N_d\cdot N_s)^{3.5} ))$ in the worst case. Although the complexity of the IASPO algorithm is exponential to $N_d\cdot N_s$, the actual complexity is usually much less than that of the worst case.
\end{proof}

%\section*{Acknowledgment}



%\bibliographystyle{IEEEtran}

\bibliographystyle{IEEEtran}
\bibliography{resource_allocation}

\begin{thebibliography}{10}
\providecommand{\url}[1]{#1}
\csname url@samestyle\endcsname
\providecommand{\newblock}{\relax}
\providecommand{\bibinfo}[2]{#2}
\providecommand{\BIBentrySTDinterwordspacing}{\spaceskip=0pt\relax}
\providecommand{\BIBentryALTinterwordstretchfactor}{4}
\providecommand{\BIBentryALTinterwordspacing}{\spaceskip=\fontdimen2\font plus
\BIBentryALTinterwordstretchfactor\fontdimen3\font minus
  \fontdimen4\font\relax}
\providecommand{\BIBforeignlanguage}[2]{{%
\expandafter\ifx\csname l@#1\endcsname\relax
\typeout{** WARNING: IEEEtran.bst: No hyphenation pattern has been}%
\typeout{** loaded for the language `#1'. Using the pattern for}%
\typeout{** the default language instead.}%
\else
\language=\csname l@#1\endcsname
\fi
#2}}
\providecommand{\BIBdecl}{\relax}
\BIBdecl

\bibitem{series2015imt}
M.~Series, ``{IMT Vision}--framework and overall objectives of the future
  development of {IMT} for 2020 and beyond,'' {Recommendation ITU}, Tech. Rep.
  M.2083, 2015.

\bibitem{xilouris2018uav}
G.~K. Xilouris, M.~C. Batistatos, G.~E. Athanasiadou, G.~Tsoulos, H.~B.
  Pervaiz, and C.~C. Zarakovitis, ``{UAV}-assisted 5{G} network architecture
  with slicing and virtualization,'' in \emph{2018 IEEE Globecom Workshops (GC
  Wkshps)}.\hskip 1em plus 0.5em minus 0.4em\relax IEEE, 2018, pp. 1--7.

\bibitem{zhang2019cellular}
S.~Zhang, H.~Zhang, B.~Di, and L.~Song, ``Cellular {UAV}-to-{X} communications:
  Design and optimization for multi-{UAV} networks,'' \emph{IEEE Transactions
  on Wireless Communications}, vol.~18, no.~2, pp. 1346--1359, 2019.

\bibitem{cui2019multi}
J.~Cui, Y.~Liu, and A.~Nallanathan, ``Multi-agent reinforcement learning-based
  resource allocation for {UAV} networks,'' \emph{IEEE Transactions on Wireless
  Communications}, vol.~19, no.~2, pp. 729--743, 2019.

\bibitem{zhao2019joint}
N.~Zhao, X.~Pang, Z.~Li, Y.~Chen, F.~Li, Z.~Ding, and M.-S. Alouini, ``Joint
  trajectory and precoding optimization for {UAV}-assisted {NOMA} networks,''
  \emph{IEEE Transactions on Communications}, vol.~67, no.~5, pp. 3723--3735,
  2019.

\bibitem{tang2019joint}
R.~Tang, J.~Cheng, and Z.~Cao, ``Joint placement design, admission control, and
  power allocation for {NOMA}-based {UAV} systems,'' \emph{IEEE Wireless
  Communications Letters}, 2019.

\bibitem{duan2019resource}
R.~Duan, J.~Wang, C.~Jiang, H.~Yao, Y.~Ren, and Y.~Qian, ``Resource allocation
  for multi-{UAV} aided {IoT} {NOMA} uplink transmission systems,'' \emph{IEEE
  Internet of Things Journal}, vol.~6, no.~4, pp. 7025--7037, 2019.

\bibitem{al2014optimal}
A.~Al-Hourani, S.~Kandeepan, and S.~Lardner, ``Optimal {LAP} altitude for
  maximum coverage,'' \emph{IEEE Wireless Communications Letters}, vol.~3,
  no.~6, pp. 569--572, 2014.

\bibitem{rouphael2009rf}
T.~J. Rouphael, \emph{RF and digital signal processing for software-defined
  radio: a multi-standard multi-mode approach}.\hskip 1em plus 0.5em minus
  0.4em\relax Newnes, 2009.

\bibitem{zhang2018energy}
H.~Zhang, B.~Wang, C.~Jiang, K.~Long, A.~Nallanathan, V.~C. Leung, and H.~V.
  Poor, ``Energy efficient dynamic resource optimization in {NOMA} system,''
  \emph{IEEE Transactions on Wireless Communications}, vol.~17, no.~9, pp.
  5671--5683, 2018.

\bibitem{xiao2018joint}
Z.~Xiao, L.~Zhu, J.~Choi, P.~Xia, and X.-G. Xia, ``Joint power allocation and
  beamforming for non-orthogonal multiple access ({NOMA}) in {5G} millimeter
  wave communications,'' \emph{IEEE Transactions on Wireless Communications},
  vol.~17, no.~5, pp. 2961--2974, 2018.

\bibitem{lee2011mixed}
J.~Lee and S.~Leyffer, \emph{Mixed integer nonlinear programming}.\hskip 1em
  plus 0.5em minus 0.4em\relax Springer Science \& Business Media, 2011, vol.
  154.

\bibitem{mosek2019}
{MOSEK ApS}, ``Introducing the {MOSEK} optimization suite 8.1.0.82,''
  \url{https://docs.mosek.com/8.1/intro/index.html}, 2019.

\bibitem{boyd2004convex}
S.~Boyd and L.~Vandenberghe, \emph{Convex optimization}.\hskip 1em plus 0.5em
  minus 0.4em\relax Cambridge university press, 2004.

\end{thebibliography}
%\begin{thebibliography}{1}
%\end{thebibliography}

\end{document}